\documentclass[11pt]{article}

\usepackage{hyperref}

\usepackage{amsmath}
\usepackage{fullpage}
\usepackage{amsfonts}
\usepackage{amssymb}
\usepackage{graphicx}
\usepackage{color}
\usepackage{theoremref}
\usepackage{times}
\usepackage{overpic}

\usepackage{authblk} 

\usepackage[nottoc]{tocbibind} 

\usepackage{tabu} 

\usepackage[multiple]{footmisc}

\newtheorem{theorem}{Theorem}

\newtheorem{lemma}[theorem]{Lemma}
\newtheorem{proposition}[theorem]{Proposition}
\newtheorem{corollary}[theorem]{Corollary}
\newtheorem{remark}[theorem]{Remark}
\newtheorem{definition}[theorem]{Definition}
\newtheorem{example}[theorem]{Example}
\numberwithin{equation}{section}
\newenvironment{proof}[1][Proof]{\noindent\textbf{#1.} }{\ \rule{0.5em}{0.5em}}

\renewcommand{\epsilon}{\varepsilon}
\newcommand{\ket}[1]{\mathop{\left|#1\right>}\nolimits}
\newcommand{\bra}[1]{\mathop{\left<#1\,\right|}\nolimits}

\newcommand{\kb}[2]{| #1\rangle\!\langle #2 |}
\newcommand{\Tra}[1]{\mathop{{\mathrm{Tr}}_{#1}}}
\newcommand{\Tr}[2]{\mathop{{\mathrm{Tr}}_{#1}} (#2) }

\def\T{\mathcal{T}}

\def\D{\mathcal{D}}

\def\U{\mathcal{U}}

\def\ZO{\{0,1\}}
\def\NN{{\mathbb N}}

\def\QCC{\mathrm{QCC}}
\def\QIC{\mathrm{QIC}}
\def\CIC{\mathrm{CIC}}
\def\CRIC{\mathrm{CRIC}}
\def\HIC{\mathrm{HIC}}

\def\RIC{\mathrm{RIC}}
\def\SCIC{\mathrm{SCIC}}
\def\SCRIC{\mathrm{SCRIC}}
\def\SHIC{\mathrm{SHIC}}
\def\HCIC{\mathrm{HCIC}}
\def\HCRIC{\mathrm{HCRIC}}
\def\HHIC{\mathrm{HHIC}}

\def\DISJ{\mathrm{DISJ}}
\def\AND{\mathrm{AND}}
\def\OR{\mathrm{OR}}

\def\PP{\mathbb{P}}
\def\bbE{\mathbb{E}}

\def\circuit{\mathbf C}

\makeatletter
\let\@copyrightspace\relax
\makeatother

\begin{document}

\title{The Flow of Information in  Interactive Quantum Protocols :  \\ the Cost of Forgetting~\footnote{A one-page abstract of this work will appear in the Proceedings of the 8th Innovations in Theoretical Computer Science conference (ITCS 2017).}}
\author[1]{Mathieu \textsc{Lauri{\`e}re}}
\author[2,3]{Dave \textsc{Touchette}}
\affil[1]{NYU-ECNU Institute of Mathematical Sciences at NYU Shanghai, email:~mathieu.lauriere@gmail.com}
\affil[2]{Institute for Quantum Computing, and Department of Combinatorics and Optimization, University of Waterloo,
email: touchette.dave@gmail.com}
\affil[3]{Perimeter Institute for Theoretical Physics}

\date{} 

\maketitle

\vspace{0.5cm}

\abstract{
In the context of two-party interactive quantum communication protocols,
we study a recently defined notion of  quantum information cost (QIC), which  possesses most of the important properties of its classical analogue, see Ref.~\cite{Tou15}. Notably, its link with amortized quantum communication complexity has been used in Ref.~\cite{BGKMT15} to prove an (almost) tight lower bound on the bounded round quantum complexity of Disjointness.
However, the only known characterization of QIC was through a notion of purification of the input state. Although this definition has the advantage to be valid for fully quantum inputs and tasks, its interpretation for classical tasks remained rather obscure.
Also, the link between this new notion and other notions of information cost for quantum protocols that had previously appeared in the literature (e.g. in Refs.~\cite{JRS03, JN14,KLLGR15}) was not clear, if existent at all. 

We settle both these issues: for quantum communication with classical inputs, we provide an alternate characterization of QIC in terms of information about the input registers, avoiding any reference to the notion of a purification of the classical input state. We provide an exact operational interpretation of this alternative characterization as the sum of the cost of transmitting information about the classical inputs and the cost of forgetting information about these inputs. 
To obtain this characterization, we prove a general lemma, the Information Flow Lemma, assessing exactly the transfer of information in general interactive quantum processes. Specializing this lemma to interactive quantum protocols accomplishing classical tasks, we are also able to demystify the link between QIC and these other previous notions of information cost in quantum protocols. Furthermore, we clarify the link between QIC and IC of classical protocols by simulating quantumly classical protocols.

Finally, we apply these concepts to argue that any quantum protocol that does not forget information solves Disjointness on $n$-bits in $\Omega (n)$ communication, completely losing the quadratic quantum speedup. This provides a specific sense in which forgetting information is a necessary feature of interactive quantum protocols in order to obtain any significant improvement over classical protocols. We also apply these concepts to prove that QIC at zero-error
is exactly $n$ for the Inner Product function, and $n (1 - o(1))$ for a random Boolean function on $n+n$ bits.

}

\newpage
\setcounter{page}{1}

\tableofcontents

\newpage
\setcounter{page}{1}

\section{Introduction}

\paragraph{Background. }

In two-party communication complexity~\cite{yao1979-distributive}, Alice and Bob receive inputs $x$ and $y$ and run an interactive communication protocol by exchanging messages in order to compute~$f(x,y)$ for some function $f$ that depends on both these inputs. Their goal is to minimize the communication cost (denoted CC and QCC respectively in the classical and the quantum settings), that is, the amount of communication (bits or qubits). This model has found numerous applications in many areas of computer science. For excellent introductions to classical and quantum communication complexities, we refer the reader to~\cite{MR1426129-KN-97-CC} and~\cite{MR1944458-DW-02-QCC} respectively.

One question that has received a lot of attention recently is whether it is possible to perform such protocols without leaking much information. 
In classical communication protocols, the information cost (IC) is defined as the information that the transcript reveals to each player about the input of the other one. In quantum communication protocols~\cite{MR1328432-Y93}, the registers are in a quantum state, which, in general, prevents the player from keeping track of the previous messages due to the no-cloning theorem. Nevertheless, the parties have quantum workspaces, where they may keep information about previous messages. The question is then to calculate how much information every new message reveals to them, given that they already know their own input and have kept some information in their quantum workspace according to the protocol.

Several notions of information cost for quantum protocols have already been used in the literature, see e.g. Refs~\cite{KlauckNTZ07, MR2050849-K-02-qprivacy, JRS03, JRS09, JN14}.
Each notion was somehow tailor-made for a specific purpose and very useful in that particular case. Nevertheless, these definitions did not seem to provide a general understanding of how information behaves in quantum communication. In Ref.~\cite{Tou15} has been introduced a general notion of Quantum Information Cost (QIC), which measures the total amount of {\it quantum} information about the inputs that is transmitted during the protocol. The corresponding notion of quantum information complexity of a function (the minimal QIC of a protocol computing the function) has been shown to exactly characterize the amortized communication complexity of that function, which is a fundamental property of the information complexity in the classical setting, see Ref.~\cite{MR2933311-BR11-amortized}. Moreover, this notion of QIC has already found multiple applications~\cite{Tou15, BGKMT15, NT16}.

However, so far the only known characterization of QIC was through a notion of purification of the input state. Although this definition has the advantage to be valid for fully quantum inputs and tasks, its interpretation for classical tasks remained rather obscure.
Also, the link between this new notion and other notions of information cost for quantum protocols that had previously appeared in the literature was not clear, if existent at all.

\paragraph{Our contributions. }

In this paper we shed a new light on the Quantum Information Cost ($\QIC$), and settle both issues described above by relating this quantity to several other natural notions of information cost, including  the classical IC, and by providing,  when the inputs are classical, a new characterization of QIC which has an operational interpretation and does not require any reference to a purification register.

The cornerstone of our work is a general lemma, that we call the Information Flow Lemma (see Lemma~\ref{lem:rdindeplb}), which precisely characterizes the transfer of information in quantum processes, run on arbitrary quantum inputs. This result then specializes to the setting we are interested in, namely quantum communication protocols.
We stress that this lemma has already found other applications besides this work, in particular to prove a lower bound on quantum information complexity of the Augmented Index function on a uniform distribution over the zeros of the function~\cite{NT16}, with  corollaries on the space complexity of quantum streaming algorithm for the $DYCK (2)$ problem of well-formed parentheses over two pairs of symbols.

We then turn our attention to quantum protocols with \textit{classical inputs}.
In this framework, even though some protocols might modify the input register, it is always possible, since the inputs are classical, to require that the players start the protocol by making a copy of their inputs and work with that copy. We call protocols such as these, where the input registers are left untouched, \emph{safe protocols}. This seemingly insignificant modification of the original protocol might drastically change the information cost. However, we prove that it can only decrease it (see Proposition~\ref{prop:safeQIC}). So it is enough to study the information cost of safe protocols when we are interested in minimizing the $\QIC$ for computing a task with classical inputs. 

When studying such quantum protocols with classical inputs, a notion of information cost (called Classical input Information Cost, or CIC) has been introduced in Ref.~\cite{KLLGR15}, where a first step was made to understand its relationship with QIC: the former is a lower bound on the latter -- that is, $\CIC \leq \QIC$. In order to complete the picture, we introduce two new notions: the Holevo Information Cost ($\HIC$), which measures how much information the players have about each other's input \textit{at the end} of the protocol (a round-by-round variant was considered in Ref.~\cite{JRS03, JN14}), and the Classical input Reverse Information Cost ($\CRIC$), which counts how much information about the inputs is \textit{forgotten} at each round by the player sending the message (this is somehow the dual under time reversal of CIC). 
 Based on our Information Flow Lemma, we give new operational interpretations to these quantities and, informally speaking, we show that they satisfy the two following very natural relationships: the Holevo information cost corresponds to the amount of classical information that was learnt and not forgotten during the protocol, while the quantum information cost captures all of the information transmitted during the protocol (what was learnt plus what was forgotten). This yields a new characterization of $\QIC$ by $\CIC$, up to a factor of $2$.
So the various notions of information cost introduced in this paper are tightly related, namely (see Propositions~\ref{thm:HicCicCric},~\ref{thm:QicCicCric} and~\ref{prop:QICleqCIC}): 
\vspace{1mm}%
\\
\noindent
\underline{Main Result 1:}
\textit{We have: $\HIC = \CIC - \CRIC,	\QIC = \CIC + \CRIC$. Moreover, $\CIC \leq \QIC \leq 2 \cdot \CIC$.
}

\vspace{2mm}%
These relationships emphasize the importance of CRIC, the cost of forgetting information.
This last quantity would always be zero in classical protocols: implicitly, classical information is always cloneable, hence players can memorize the whole history of the protocol and never forget information.
To understand the link with quantum protocols forgetting information, we introduce a model of classical reversible computing, endowing classical protocols with the ability to forget information. We show that this feature can only increase their information cost, and, as such, forgetting information is somehow a wasteful phenomenon that should be avoided in the context of classical communication (see Theorem~\ref{thm:simu-classical-NF}). However, in quantum protocols, cloning is not possible in general. This raises the question whether the property of forgetting information is only costly and should still be avoided in some sense. We answer this in the negative: forgetting information is absolutely necessary to obtain the quantum communication improvement allowed for computing certain functions. Indeed, if no information is forgotten in a quantum protocol, then $\QIC = \HIC$ is formally very similar to IC, and  the continuity in the input distribution has no round dependence, as in the classical case. Thus, the round dependence in this continuity bound for general quantum protocols that do forget information~\cite{BGKMT15} can be understood as being due to the fact that the \emph{same information} is forgotten and transmitted multiple times.
With this observation, we prove that any quantum protocol for Disjointness that does not forget information has linear quantum communication complexity (see Theorem~\ref{thm:disj-NF-n}). Hence, quantum protocols that do not forget information cannot obtain the quadratic quantum speed-up for the Disjointness function~\cite{MR2322514-AA-05-spatialsearch}, and this ability of quantum protocol to forget information is an essential feature of interactive quantum communication, not just some oddity we can get around. This can be summarized as follows:
\vspace{1mm}%
\\
\noindent
\underline{Main Result 2 :}
\textit{Forgetting information is useless in a classical reversible setting, but it is unavoidable in the quantum setting: it is a necessary feature of interactive quantum protocols to get significant communication improvement over classical protocols.}

\vspace{2mm}%
This important distinction shows that the flow of information behaves quite differently in the classical and in the quantum setting. However, the classical communication complexity is always lower bounded by the quantum communication complexity: quantum messages can simulate classical ones. We can ask the same question in terms of information: 
is it always possible to quantumly simulate classical messages \emph{while maintaining the information cost}?
Our next main result provides a positive answer.
We show that to any classical protocol $\Pi_C$ corresponds a quantum simulation protocol $\Pi_Q$ satisfying $\QCC (\Pi_Q) = CC (\Pi_C)$, $\QIC(\Pi_Q, \mu) = IC (\Pi_C, \mu)$ for any input distribution $\mu$, and implementing the same input-output channel $\Pi_Q = \Pi_C$. The main issue we deal with is the pure state quantum simulation of private randomness without altering the information cost (see Lemma~\ref{lem:simu-quantum}).
\vspace{1mm}%
\\
\noindent
\underline{Main Result 3 :}
\textit{ For any classical protocol, there exists a quantum protocol with the same input-output behaviour, and with communication and information costs smaller than the classical protocol.
}

\vspace{2mm}%
This result lets us conclude the paper with one more application. For the Inner Product function, QIC at zero-error over the uniform distribution is exactly $n$; a similar lower bound of $n (1 - o(1))$ holds for a random Boolean function  on $n+n$ bits.
Further using the quantum simulation of classical protocols mentionned above together with the fact that classical IC is continuous at zero-error~\cite{MR3210776-infoExactComm}, this shows that, in the limit when the error $\epsilon$ goes to $0$, IC of such a random Boolean function is not only $\Omega(n)$~\cite{MR3003574-discLBIC, kerenidis2015lower}, but is precisely $n (1 - o(1))$ (such a tight bound for the IC of Inner Product was known from Ref.~\cite{braverman2013information}).

\paragraph{Outline of the paper. }
This paper is structured as follows. After some preliminaries (Section~\ref{sec:ifl-prelim}), we state and prove our Information Flow Lemma  (Section~\ref{sec:infoflowlemma}). In Sections~\ref{sec:safe-quantum} and~\ref{sec:cost-forget-new-charac-QIC}, we prove our results on  safe quantum protocols, and then introduce CRIC, HIC and multiple other quantum notions of information cost (a table is provided in Appendix~\ref{sec:inflow-B-table-IC} to keep track of definitions and relationships). For the sake of comparison, in Section~\ref{sec:clforget} we define IC in a classical reversible computation paradigm and show that forgetting information is wasteful. In contrast, we prove in Section~\ref{sec:disj} that there is no quantum communication speed-up for Disjointness when the quantum protocols are not allowed to forget information.
 Then, we show how to simulate quantumly classical protocols in Section~\ref{sec:quantum-simu-classical}. Finally we prove our results on Inner Product and random Boolean functions  (Section~\ref{sec:IPrandfct}).

\section{Preliminaries: Quantum Communication and Information}
\label{sec:ifl-prelim}

\paragraph{Quantum Communication Model.} 
\label{sec:qccmodel}

Quantum communication complexity was introduced by Yao in Ref.~\cite{MR1328432-Y93}. The model we use here is closer to the one of Cleve and Buhrman~\cite{CB97}, with pre-shared entanglement, but we allow the players to communicate with quantum messages.
In this model, an $r$-round protocol $\Pi$ for a given classical task from input registers $A_{in} = X$, $B_{in} = Y$ to output registers $A_{out}$, $B_{out}$ is 
defined by a sequence of isometries $U_1$, $\cdots$, $U_{M + 1}$ along with a 
pure state $\psi \in \D (T_A^{in}  T_B^{in})$ shared between Alice and Bob,
for arbitrary finite dimensional registers $T_A^{in}$, $T_B^{in}$: the pre-shared entanglement. Above, $\D(\mathcal A)$ is the set of all unit trace, positive semi-definite linear operators mapping $\mathcal A$ into itself.
See Refs~\cite{Watrous15, Wilde13}.
We need $r+1$ isometries in order to have $r$ messages since a first isometry is applied before the first message
is sent and a last one after the final message is received.
In  the case of even $r$, for appropriate finite 
dimensional quantum memory registers $A_1$, $A_3$, $\cdots$, $A_{r - 1}$, $A^\prime$ held by Alice, $B_2$, $B_4$, $\cdots$, $B_{r - 2}$, $B^\prime$ 
held by Bob, and quantum communication registers $C_1$, $C_2$, $C_3$, $\cdots$, $C_r$ 
exchanged by Alice and Bob, we have
$U_1 \in \U(A_{in}  T_A^{in}, A_1  C_1)$, 
$U_2 \in \U(B_{in}  T_B^{in}  C_1, B_2  C_2)$, 
$U_3 \in \U(A_1  C_2, A_3  C_3)$, 
$U_4 \in \U(B_2  C_3, B_4  C_4)$, $\cdots$ , 
$U_{r} \in \U(B_{r - 2}  C_{r - 1},   B_{out}  B^\prime  C_{r})$,
$U_{r  + 1} \in \U(A_{r - 1}  C_r, A_{out}  A^\prime)$, where $\U(\mathcal A, \mathcal B)$ is the set of unitary channels from $\mathcal A$ to $\mathcal B$ :
see Figure~\ref{fig:int_mod}.
We adopt the convention that, at the outset, $A_0 = A_{in} T_A^{in}$, $B_0 = B_{in} T_B^{in}$, 
for odd $i$ with $1 \leq i < r$, $B_i = B_{i-1}$, for even $i$ with $1 < i \leq r$, $A_i = A_{i-1}$ and also $B_r = B_{r + 1} = B_{out} B^\prime$, and $A_{r+1} = A_{out} A^\prime$. In this way, after application of $U_i$, Alice holds register $A_i$, Bob holds register $B_i$ and the communication register is $C_i$. 
In the case of an odd number of messages $r$, the registers corresponding to $U_r$, $U_{r+1}$ are changed accordingly.
We slightly abuse notation and also write $\Pi$ to denote the channel from registers $ A_{in} B_{in}$ to $A_{out} B_{out}$ implemented by the protocol, i.e.~for any 
input distribution $\mu$ on $XY$ and $\rho_\mu$ encoding $\mu$ on input registers $A_{in} B_{in}$,
\begin{align}
\Pi (\rho_\mu) =
\Tr{A^\prime B^\prime }{U_{M  + 1} U_M \cdots U_2 U_1 (\rho_\mu \otimes \psi)}.
\end{align}

Note that the $ A^\prime $ and $ B^\prime $ registers are the final memory registers that are being discarded at the end of the protocol by Alice and Bob, respectively.

Recall that for a given state,  all purifications are related by isometries on the purification registers. For classical input registers $XY$ distributed according to $\mu$, we consider a canonical purification $\ket{\rho_\mu}^{X R_X Y R_Y}$ of $\rho_\mu^{A_{in} B_{in}}$, with
\begin{align}
\ket{\rho_\mu}^{X R_X Y R_Y} = \sum_{x, y} \sqrt{\mu (x, y)} \ket{xxyy}^{X R_X Y R_Y}.
\end{align}
We then say that the purifying registers $R_X R_Y$ contain \emph{quantum copies} of $XY$.
Then, the state at round $i$,
\begin{align}
\rho_i^{X R_X Y R_Y A_i B_i C_i} = U_i \cdots U_1 (\rho^{X R_X Y R_Y} \otimes \psi^{T_A^{in} T_B^{in}})
\end{align}
is pure. Also, we require that the final marginal state $\Pi (\rho^{A_{in} B_{in} R_X R_Y})$
on $R_X R_Y A_{out} B_{out}$ is classical. We say that a protocol $\Pi$ solves a function $f$ with error $\epsilon$ with respect to input distribution $\mu$ if $\Pr_\mu [\Pi (x, y) \not= f(x, y)] \leq \epsilon$, and we say $\Pi$ solves $f$ with error $\epsilon$ if $\max_{(x, y)} \Pr[\Pi (x, y) \not= f(x, y)] \leq \epsilon$.

We also make use of the notion of a control-isometry: it is an isometry acting on a classical-quantum register by leaving the content of the classical register unchanged. Such a classical register is called a control-register.

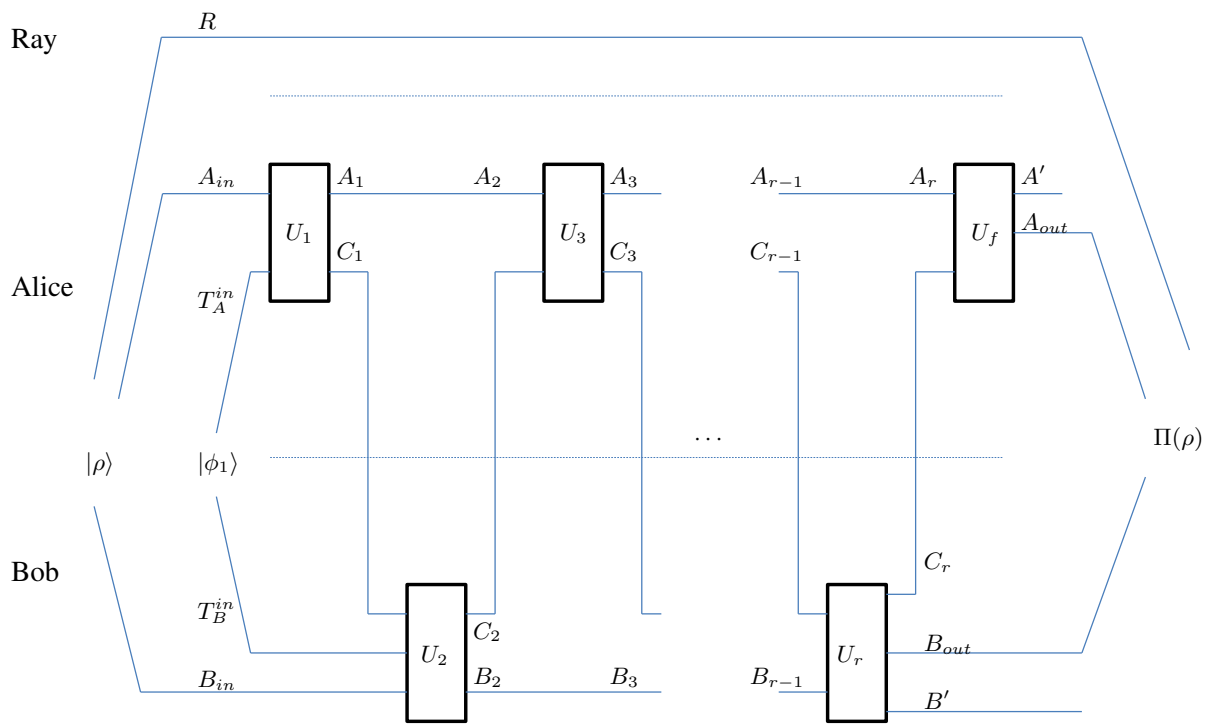
\begin{figure}
\begin{overpic}[width=1\textwidth]{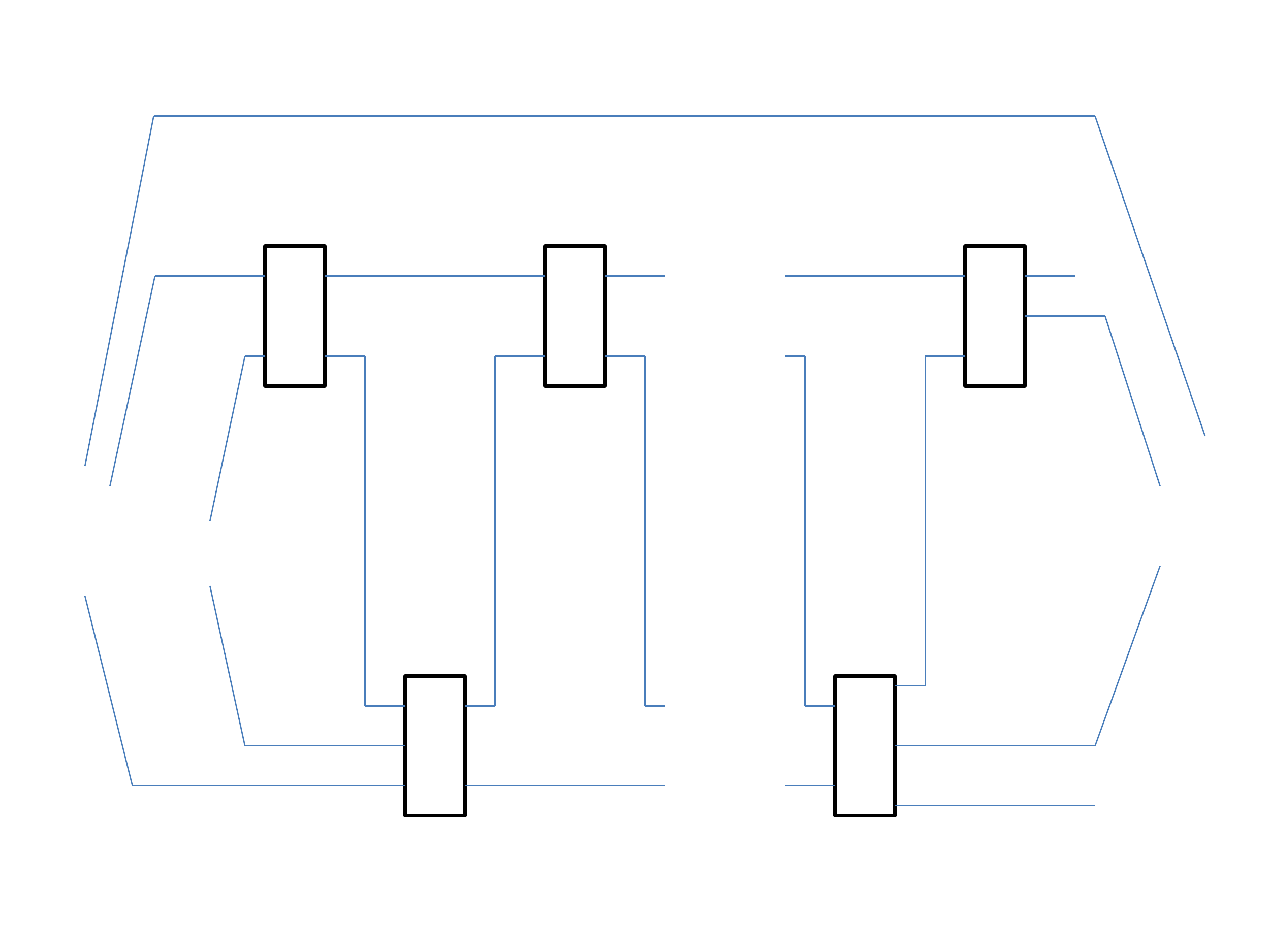}
  \put(0,65){Ray}
  \put(0,45){Alice}
  \put(0,22){Bob}
  \put(6,31){\footnotesize $\ket{\rho}$}
  \put(15,66.5){\footnotesize$R$}
  \put(15,54){\footnotesize$A_{in}$}
  \put(15,13.5){\footnotesize$B_{in}$}
  \put(15,44){\footnotesize$T_A^{in}$}
  \put(15,19){\footnotesize$T_B^{in}$}
  \put(22.1,49.5){\footnotesize$U_1$}
  \put(15,31){\footnotesize$\ket{\phi_1}$}
  \put(26.2,54){\footnotesize$A_1$}
  \put(26.2,48){\footnotesize$C_1$}
  \put(33,15.5){\footnotesize$U_2$}
  \put(37.2,54){\footnotesize$A_2$}
  \put(37.2,17.4){\footnotesize$C_2$}
  \put(37.2,13.7){\footnotesize$B_2$}
  \put(44.2,49.5){\footnotesize$U_3$}
  \put(48.2,54){\footnotesize$A_3$}
  \put(48.2,48){\footnotesize$C_3$}
  \put(48.2,13.7){\footnotesize$B_3$}
  \put(55,33){\footnotesize$\cdots$}
  \put(59.5,54){\footnotesize$A_{r-1}$}
  \put(59.5,48){\footnotesize$C_{r-1}$}
  \put(59.5,13.7){\footnotesize$B_{r-1}$}
  \put(66.5,15.5){\footnotesize$U_{r}$}
  \put(72.3,54){\footnotesize$A_r$}
  \put(73.5,23){\footnotesize$C_r$}
  \put(73.5,16.4){\footnotesize$B_{out}$}
  \put(73.5,11.7){\footnotesize$B^{\prime}$}
  \put(77.3,49.5){\footnotesize$U_{f}$}
  \put(81.3,54){\footnotesize$A^\prime$}
  \put(81.3,50.3){\footnotesize$A_{out}$}
  \put(92,33){\footnotesize$\Pi (\rho)$}
\end{overpic}
  \caption{Depiction of a quantum protocol in the interactive model, adapted from the long version of~\cite[Figure 1]{Tou15}.}
  \label{fig:int_mod}
\end{figure}

\paragraph{Quantum Information Cost.} The main quantity of interest in this work is the quantum information cost, as introduced in~\cite{Tou15}. In quantum communication protocols, there is no clear notion of a transcript, so this definition counts how much information is exchanged in each round. In the sequel, we denote the Von Neumann entropy by $H$, and for a tripartite state $\rho^{ABC}$, we denote  the conditional quantum mutual information (CQMI) between $A$ and $B$ conditioned on $C$ by $I(A:B|C) = H(A,C) + H(B,C) - H(C) - H(A,B,C)$. We will make use of many properties of CQMI, among which the following.
\begin{lemma}
\label{prelim:lem:CQMIfacts}
	If $\rho = \rho^{ABC}$ and $\sigma = \sigma^{DEF}$ are two states on distinct registers, then
	$$
		I(AD ; BE|CF)_{\rho \otimes \sigma} = I(A ; B| C)_{\rho} + I(D ; E | F)_{\sigma}.
	$$
	If $\rho = \rho^{ABCD} = \sum_{c} p(c) \kb{c}{c} \otimes \rho_c^{ABD} $ is a classical-quantum state with classical register $C$, then
	$$
		I(A:B|CD)_{\rho} = \bbE_c\left[ I(A:B|D)_{\rho_c} \right].
	$$
	If $\rho = \rho^{ABCD}$ is a pure state, then
	$$
		I(A ; B|C)_{\rho} = I(A ; B| D)_{\rho}.
	$$
\end{lemma}
Let us recall the definition of quantum information cost introduced in~\cite{Tou15}.
\begin{definition}
\label{prelim:def:QIC}
For a protocol $\Pi$ and an input distribution $\mu$, 
we define the \emph{quantum information cost} of $\Pi$ on input $\mu $ as
\begin{align*}
	\QIC (\Pi, \rho) &= \sum_{i \geq 1,\ odd} I(C_i; R_X R_Y | B_{i}) + \sum_{i \geq 1,\ even}  I(C_i; R_X R_Y | A_{i}).
\end{align*}
For any function $f$, any input distribution $\mu$,  and any $\epsilon>0$
	\begin{align}\label{eq:defQIComplexity}
		\QIC (f, \mu, \epsilon) = \inf_{\Pi} \QIC (\Pi, \mu)
	\end{align}
	where the infimum is over the protocols $\Pi$ computing $f$ with error $\epsilon$ w.r.t $\mu$.
\end{definition}
This quantity has many nice properties (see \cite{Tou15, BGKMT15}); in particular it characterizes the (quantum) amortized communication complexity.
We stress that the definition is independent of the choice of purification.

\paragraph{Discussion about compression. }
Some previous notions of information cost for quantum protocols (e.g. in Refs.~\cite{JRS03, JN14,KLLGR15}) were more similar in spirit to classical input information cost than to quantum information cost. Our results shed new light on why these previous definitions were restricted to compression results for a single round. In the first round, Alice does not yet possess any information on Bob's input (aside from what she can infer from her own input). 
For one-round protocols, it is then immaterial whether one uses classical input information cost or quantum information cost.
But then in subsequent rounds, generally Alice has in her registers some information about Bob's input. It is then possible for her to forget information while sending a message. We can even construct a protocol where, at the third round, Bob does not learn anything whereas Alice forgets a lot of information. For such a round of communication, the previous definitions of information cost, e.g. CIC introduced in Ref.~\cite{KLLGR15}, would evaluate to $0$ whereas QIC would be large. Thus, it is impossible to compress such a quantum message down to its CIC, that is, almost at no cost, while keeping, in a round-by-round fashion, the overall state of the protocol  almost equivalent  to that in the original protocol. Indeed, we know from our developments that to forget information we must invest communication.
As a consequence, we see that for quantum protocols, it is important to take into account the cost of forgetting information.

The purification register used in the definition of QIC possibly appears artificial when considering classical inputs. In this direction, we prove below  (see Section~\ref{sec:cost-forget-new-charac-QIC}) an arguably more natural characterization (at least from a classical correlation point of view) of each term in the quantum information cost as the sum of how much information about his own input a party is sending plus how much information about the other party's input he is forgetting.
However, we argue that there is still virtue in taking the purification of the classical input viewpoint.
Firstly, it enables to keep track of a global pure state, which in many situations is a remarquably powerful viewpoint.
Secondly and more fundamentally, the purification viewpoint has a nice operational interpretation through the task of quantum state redistribution, which is useful when aiming at compression results.
Indeed, at any point of the interactive protocol, the pure quantum state can be seen as a $4$-partite state $\rho^{A_R A_S M R}$ consisting of the receiver's and the sender's private registers  ($A_R$ and $A_S$ respectively), the message  register $M$ and a purification register  $R$. Then, each term in QIC is of the form $I(R;M | A_R)$, that is, the mutual information between the message and the inaccessible purification register, conditioned on the receiver's side information. Such an expression is known~\cite{DY08, YD09} to quantify the cost of redistributing the message register while maintaining correlations with the receiver's and the sender's private registers as well as the environment.
The $\CIC$ terms can also be given such an operational significance for the information about the sender's input that a message contains. However, this viewpoint breaks down for the information that is forgotten (see the operational interpretation given at Section~\ref{sec:cost-forget-new-charac-QIC}). Indeed, to measure the amount of information being forgotten, we condition on the sender's side information for sending information about the receiver's input. This term would be hard to account for in a compression viewpoint (unless we think of messages going backward). Hence, we think that the purification viewpoint remains appropriate for compression purposes.




\section{Information Flow Lemma}
\label{sec:infoflowlemma}

In this section, we state and prove the \emph{Information Flow Lemma} (see Lemma~\ref{lem:rdindeplb} below), which  allows to keep track exactly of the flow of quantum information in an interactive protocol and is key to much of our further developments.
Moreover, it gives a lower bound on QIC that does not depend on the number of round (see Corollary~\ref{coro:roundIndepLB}), and is used, among other things, to give an exact meaning to the cost of forgetting in interactive quantum protocols.
We present here a quite general version of this result. However, we stress that a more limited version, that is still sufficient to obtain a lower bound on QIC, has already found some applications; see Ref.~\cite{NT16}.

Let us consider the more general framework of bipartite interactive 
quantum processes, of which the model of quantum communication complexity defined in Section~\ref{sec:ifl-prelim} is a special case.
This general framework modelizes a discretized quantum process in which there is 
interaction between two distinct, localized parties, and local 
evolution at each time step.

In more details, Alice and Bob start in a joint state $\rho_0^{\bar{A}_0 \bar{B}_0}$, 
for which we consider an arbitrary 
extension $\rho_0^{\bar{A}_0 \bar{B}_0 \bar{E} \bar{F}}$ 
(such that $\Tr{\bar{E} \bar{F}}{\rho^{\bar{A}_0 \bar{B}_0 \bar{E} \bar{F}}} 
= \rho^{\bar{A}_0 \bar{B}_0}$). 
The process runs for $r+1$ rounds, 
with $\rho_i$ the state in round $i$, 
 registers $\bar{A}_i$, $\bar{B}_i$, $\bar{C}_i$ 
and $\bar{D}_i$ in each round, with $\bar{C}_0$, $\bar{D}_0$, 
$\bar{C}_{r+1}$ and $\bar{D}_{r+1}$ being trivial registers in 
the $0$-th and $r+1$-th round, initially and at the end of the 
process. In round $i$, for $1 \leq i \leq r$, after being generated 
by Alice, register $\bar{C}_i$ gets communicated from Alice to Bob, 
and, after being generated by Bob, register $\bar{D}_i$ gets 
communicated from Bob to Alice. Register $\bar{A}_i$ is a 
quantum memory register held by Alice, and register $\bar{B}_i$ is a 
quantum memory register held by Bob. The evolution is through local 
isometries $U_i = U_i^{\bar{A}_{i-1} \bar{D}_{i - 1} \rightarrow \bar{A}_i \bar{C}_i}$ 
on Alice's side and $V_i = V_i^{\bar{B}_{i-1} \bar{C}_{i - 1} \rightarrow \bar{B}_i \bar{D}_i}$ 
on Bob's side: $\rho_i^{\bar{A}_i \bar{B}_i \bar{C}_i \bar{D}_i \bar{E} \bar{F}} = (U_i \otimes V_i) \rho_{i-1}^{\bar{A}_{i-1} \bar{B}_{i-1} \bar{C}_{i-1} \bar{D}_{i-1} \bar{E} \bar{F}}$.

\begin{figure}
\begin{overpic}[width=1\textwidth]{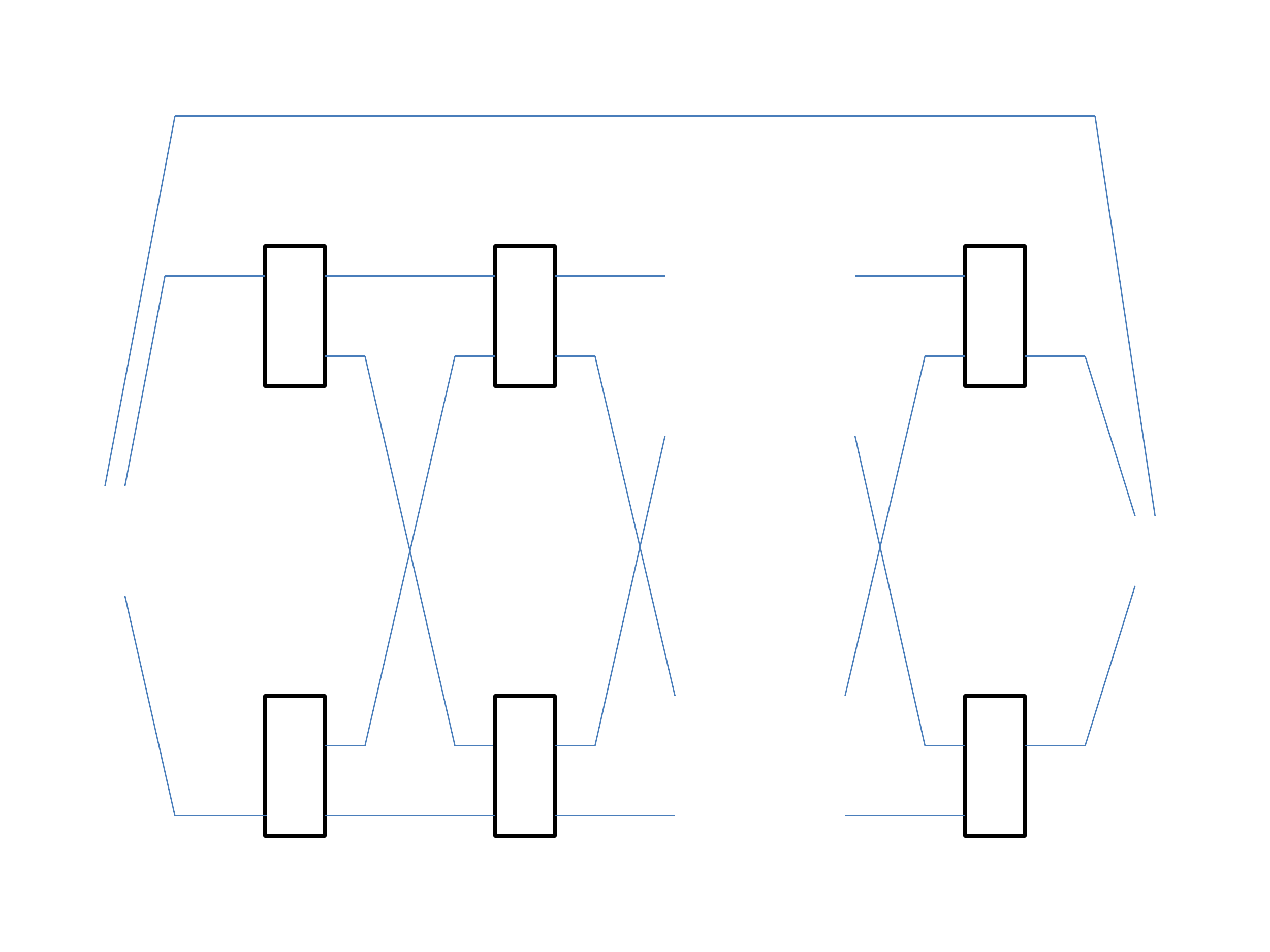}
  \put(0,65){Extension}
  \put(0,45){Alice}
  \put(0,22){Bob}
  \put(8,31){\footnotesize $\rho_0$}
  \put(15,66.5){\footnotesize$\bar{E} \bar{F}$}
  \put(15,54){\footnotesize$\bar{A}_{0}$}
  \put(15,8.5){\footnotesize$\bar{B}_{0}$}
  \put(22.1,49.5){\footnotesize$U_1$}
  \put(22.1,14){\footnotesize$V_1$}
  \put(26.2,54){\footnotesize$\bar{A}_1$}
  \put(26.2,48){\footnotesize$\bar{C}_1$}
  \put(26.2,13.5){\footnotesize$\bar{D}_1$}
  \put(26.2,8.5){\footnotesize$\bar{B}_1$}
  \put(36.2,48){\footnotesize$\bar{D}_1$}
  \put(36.2,13.5){\footnotesize$\bar{C}_1$}
  \put(40.2,49.5){\footnotesize$U_2$}
  \put(40.2,14){\footnotesize$V_2$}
  \put(44.3,54){\footnotesize$\bar{A}_2$}
  \put(44.3,48){\footnotesize$\bar{C}_2$}
  \put(44.3,13.7){\footnotesize$\bar{D}_2$}
  \put(44.3,8.5){\footnotesize$\bar{B}_2$}
  \put(58,33){\footnotesize$\cdots$}
  \put(70.8,54){\footnotesize$\bar{A}_{r-1}$}
  \put(70.8,48){\footnotesize$\bar{D}_{r-1}$}
  \put(70.8,13.7){\footnotesize$\bar{C}_{r-1}$}
  \put(70.8,8.5){\footnotesize$\bar{B}_{r-1}$}
  \put(77.3,49){\footnotesize$U_{r}$}
  \put(77.3,14){\footnotesize$V_{r}$}
  \put(81.3,48){\footnotesize$\bar{A}_{r}$}
  \put(81.3,13.7){\footnotesize$\bar{B}_{r}$}
  \put(90,31){\footnotesize $\rho_r$}
\end{overpic}
  \caption{Depiction of an interactive quantum process, adapted from the long version of~\cite[Figure 1]{Tou15}.}
  \label{fig:int_process}
\end{figure}
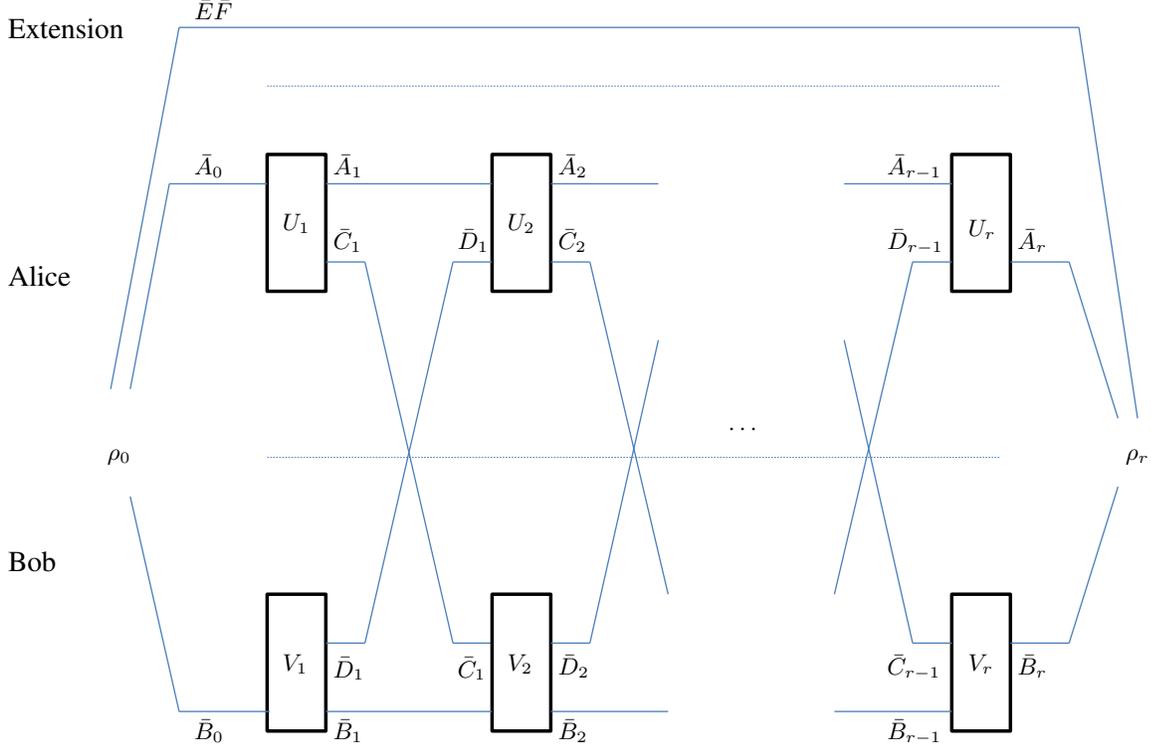

Registers $\bar{E} \bar{F}$ are left untouched throughout, and can be thought of in the following way: we want to measure how much information Bob knows about $\bar{E}$ from the point of view of someone who knows $\bar{F}$. We get the following exact characterization of the flow of information from this point of view.

\begin{lemma}
\label{lem:rdindeplb}
(Information Flow Lemma)
	Given an interactive quantum process as defined above, the following holds:
\begin{align*}
	 I(\bar{E} ; \bar{B}_{r+1} | \bar{F})_{\rho_{r+1}} - I (\bar{E} ; \bar{B}_0 | \bar{F})_{\rho_0} = \sum_{i=1}^r (I(\bar{E} ; \bar{C}_i | \bar{F} \bar{B}_i )_{\rho_i} - I(\bar{E} ; \bar{D}_i | \bar{F} \bar{B}_i )_{\rho_i} ).
\end{align*}
\end{lemma}

\begin{proof}
	
We keep track of the flow of information using the chain rule and local isometric invariance of CQMI:

\begin{align*}
	I(\bar{E} ; \bar{B}_{r+1}  | \bar{F} ) & = I (\bar{E} ; \bar{B}_{r} \bar{C}_r | \bar{F})  \\
					& = I (\bar{E} ;  \bar{B}_{r} | \bar{F} ) +  I (\bar{E} ;  \bar{C}_r | \bar{F} \bar{B}_{r})   
					 + \big( I (\bar{E} ;  \bar{D}_r | \bar{F} \bar{B}_{r}) - I (\bar{E} ;  \bar{D}_r | \bar{F} \bar{B}_{r}) \big) \\
					& = \big( I (\bar{E} ;  \bar{B}_r | \bar{F}) + I (\bar{E} ;  \bar{D}_r | \bar{F} \bar{B}_{r}) \big)  
					+ I (\bar{E} ;  \bar{C}_r | \bar{F} \bar{B}_r) - I (\bar{E} ;  \bar{D}_r | \bar{F} \bar{B}_{r}) \\
					& =  I (\bar{E} ;  \bar{B}_r \bar{D}_r | \bar{F})  
					+ I (\bar{E} ;  \bar{C}_r | \bar{F} \bar{B}_r) - I (\bar{E} ;  \bar{D}_r | \bar{F} \bar{B}_{r}) \\
					& =  I (\bar{E} ;  \bar{B}_{r-1} \bar{C}_{r-1} | \bar{F}) 
					 + I (\bar{E} ;  \bar{C}_r | \bar{F} \bar{B}_r) - I (\bar{E} ;  \bar{D}_r | \bar{F} \bar{B}_{r}).
\end{align*}
Applying recursively the same argument leads to
\begin{align*}
I(\bar{E} ; \bar{B}_{r+1}  | \bar{F} ) 
					& =  I (\bar{E} ;  \bar{B}_{1} \bar{C}_{1} | \bar{F}) 
					 + \sum_{i = 2}^r  \big( I (\bar{E} ;  \bar{C}_i | \bar{F} \bar{B}_i) - I (\bar{E} ;  \bar{D}_i | \bar{F} \bar{B}_{i}) \big)
					\\
					& =  I (\bar{E} ;  \bar{B}_{1} \bar{D}_1 | \bar{F}) 
					 + I (\bar{E} ;   \bar{C}_{1} | \bar{F} \bar{B}_{1}) - I (\bar{E} ;   \bar{D}_{1} | \bar{F} \bar{B}_{1})
					\\
					& \qquad + \sum_{i = 2}^r  \big( I (\bar{E} ;  \bar{C}_i | \bar{F} \bar{B}_i) - I (\bar{E} ;  \bar{D}_i | \bar{F} \bar{B}_{i}) \big)
					\\
					& =  I (\bar{E} ;  \bar{B}_{0} | \bar{F}) 
					+ \sum_{i = 1}^r  \big( I (\bar{E} ;  \bar{C}_i | \bar{F} \bar{B}_i) - I (\bar{E} ;  \bar{D}_i | \bar{F} \bar{B}_{i}) \big).
\end{align*}
We get the desired result by rearranging terms.
\end{proof}

In the remainder of this work, we are concerned with quantum communication protocols as defined in Section~\ref{sec:ifl-prelim}, for which an easy corollary of the Information Flow Lemma is as follows. A similar result holds for Alice.

\begin{corollary}
\label{lem:rdindeplb-coro}
	Given a protocol $\Pi$, an input distribution $\mu$ and any extension $\rho_0^{A_{in} B_{in} E_1 E_2}$
satisfying~:
$\Tr{E_1 E_2}{\rho_0^{A_{in} B_{in} E_1 E_2}} = \rho_\mu^{A_{in} B_{in}}$,
\begin{align*}
	I (E_1 ; B^\prime B_{out} | E_2)_{\rho_{r+1}} - I(E_1 ; B_{in} | E_2)_{\rho_0} = \sum_{i~odd} I(E_1 ; C_i | E_2 B_i)_{\rho_i} - \sum_{i~even} I (E_1 ; C_i | E_2 B_i)_{\rho_i}.
\end{align*}
\end{corollary}

Combining the above result and a similar one holding for Alice, we get the following lower bound on quantum information cost, stated as a sum of differences between the amount of correlations of reference registers with the output and the input.

\begin{corollary}\label{coro:roundIndepLB}
	Given a protocol $\Pi$, an input distribution $\mu$ 
and any two extensions $\rho_{0,B}^{A_{in} B_{in} E_1 E_2}$, $\rho_{0,A}^{A_{in} B_{in} F_1 F_2}$
satisfying~:
$\Tr{E_1 E_2}{\rho_{0, A}^{A_{in} B_{in} E_1 E_2}} = \rho_\mu^{A_{in} B_{in}}$,
$\Tr{F_1 F_2}{\rho_{0, B}^{A_{in} B_{in} F_1 F_2}} = \rho_\mu^{A_{in} B_{in}}$,
the following holds:
\begin{align*}
	\QIC (\Pi, \rho) & \geq  I (F_1 ; A_{out} A^\prime | F_2) - I (F_1 ; A_{in} | F_2) \\
			& \quad + I (E_1 ; B_{out} B^\prime | E_2) - I (E_1 ; B_{in} | E_2).
\end{align*}
\end{corollary}



\section{Making Safe Copies of the Inputs}
\label{sec:safe-quantum}

In this section, we show that making safe copies of classical inputs at the outset of a quantum protocol never increases its quantum information cost.
So, when studying the quantum information complexity of a function, it is always possible to assume that protocols do not change the input registers.

Following Ref.~\cite{JRS03}, we introduce the notion of safe copies and safe protocols.
\begin{definition}[Safe protocol]
	Recall that, in a quantum communication protocol implementing a classical task, players receive initial classical data in some quantum input registers. We say that such a protocol is \emph{safe} if the players only use these input registers as control registers.
\end{definition}
Note that for quantum protocols, making  a local copy of the classical input does not change the quantum communication cost. However it is not obvious from definition that the same property should be true for the information cost. Let us make this question more precise by associating to every protocol another protocol, which is safe.

\paragraph{Safe Version of a Protocol.} Consider any protocol $\Pi$. We define a safe version of $\Pi$ as follows. Let $\Pi^\prime$ be the protocol in which Alice and Bob first make a coherent (safe) copy of their respective inputs $X, Y$ at the outset of the protocol into safe registers $X^\prime, Y^\prime$, and then run $\Pi$ while using $X^\prime$ and $Y^\prime$ as inputs. Recall that there are also coherent copies held in purification registers $R_X, R_Y$. That is, on input distribution $\mu$, we denote as $\rho_\mu^{XY}$ the state
\begin{align}
	\rho_\mu^{XY} = \sum_{x, y} \mu (x, y) \kb{x}{x}^X \otimes \kb{y}{y}^Y,
\end{align}
and we consider a purification of the form
\begin{align}
	\ket{\rho_\mu}^{XY R_X R_Y} = \sum_{x, y} \sqrt{\mu (x, y)} \ket{x}^{X} \ket{y}^{Y} \ket{x}^{R_X} \ket{y}^{R_Y}.
\end{align}
In the protocol $\Pi^\prime$, the registers $X$, $Y$ are then left untouched for the remainder of the protocol, which is identical to protocol $\Pi$ acting on input registers $X^\prime$, $Y^\prime$ after such copies are made. We want to show that the quantum information cost of $\Pi^\prime$ is never greater than that of $\Pi$. More formally, define the isometries
\begin{align}
	U_X^{X \rightarrow X X^\prime} &= \sum_{x \in X} \ket{x}^{X} \ket{x}^{X^\prime} \bra{x}^X, \\
	U_Y^{Y \rightarrow Y Y^\prime} &= \sum_{y \in Y} \ket{y}^{Y} \ket{y}^{Y^\prime} \bra{y}^Y.
\end{align}
Then the safe protocol $\Pi^\prime$ is defined from $\Pi$ by:
\begin{enumerate}
	\item applying $U_X$ and then $U_1$ acting on $X^\prime$ on Alice's side in the first round,
	\item applying $U_Y$ and then $U_2$ acting on $Y^\prime$ on Bob's side in the second round,
	\item running $U_i$ in round $i$ for $i \geq 3$.
\end{enumerate}
This does not change the classical input/output behavior of the protocol.
If we think of acting $U_Y$ before $U_1$, this does not change the value of any QIC term, and we get state
\begin{align}
\ket{\rho_\mu^\prime}^{X X^\prime R_X Y Y^\prime R_Y} & = (U_X^{X \rightarrow X X^\prime} \otimes U_Y^{Y \rightarrow Y Y^\prime}) \ket{\rho_\mu}^{X R_X Y R_Y} \\
	& = \sum_{x, y} \sqrt{\mu (x, y)} \ket{xxxyyy}^{X X^\prime R_X Y Y^\prime R_Y}
\end{align}
at the outset of protocol $\Pi^\prime$.
We then show that making such safe copies does not increase the QIC of a protocol.

\paragraph{Making Safe Copies can only Decrease QIC of a Protocol.} It turns out that $\QIC(\Pi, \mu)$ and $\QIC(\Pi^\prime, \mu)$ can be very different. Let us illustrate this point with a simple example.
\begin{example}
Consider an input distribution $\mu$ such that $X$ is uniformly distributed, and $Y=X$. Consider a protocol in which Alice directly sends her input to Bob. Then the costs are
\begin{align}
	\QIC (\Pi, \mu) & = I(X: R_X R_Y | Y)_{\rho_\mu} \\
			& = I(X: R_X R_Y )_{\rho_\mu} \\
			& = H(X)_{\rho_\mu} \\
			& = \lg |X|, \\
\hbox{whereas } \quad \QIC (\Pi^\prime, \mu) & = I(X^\prime: R_X R_Y | Y^\prime Y)_{\rho_\mu^\prime} \\
			& = 0,
\end{align}
in which we used for $\QIC (\Pi^\prime, \mu)$ that all registers are classical once $X$ is traced out along with the fact that $X=Y$, similarly for $I(X: R_X R_Y )_{\rho_\mu}$ and tracing out $Y$, and finally, since $\rho_\mu^{X Y R_X R_Y}$ is 
pure, $I(X: R_X R_Y | Y)_{\rho_\mu} = I(X: R_X R_Y )_{\rho_\mu}$.
\end{example}

This phenomenon might occur even when there is no correlation between $X$ and $Y$, as shown by the following example.
\begin{example}\label{ex:forget}
Consider an input distribution $\mu$ such that $X$ and $Y$ are distributed independently and uniformly. Consider a protocol in which Alice directly sends her input to Bob. Then the costs are
\begin{align}
	\QIC (\Pi, \mu) & = I(X: R_X R_Y | Y)_{\rho_\mu} \\
			& = I (X ; R_X)_{\rho_\mu} \\
			& = 2H(X)_{\rho_\mu} \\
			& = 2 \lg |X|, \\
\hbox{whereas } \quad \QIC (\Pi^\prime, \mu) & = I(X^\prime: R_X R_Y | Y^\prime Y)_{\rho_\mu^\prime} \\
			& = I (X^\prime ; R_X)_{\rho_\mu^\prime} \\
			& = H(X^\prime)_{\rho_\mu^\prime} \\
			& = \lg |X|,
\end{align}
where we used that $\rho_\mu^{X R_X}$ is a pure state whereas $\rho_\mu^\prime$ is classical on $X^\prime R_X$ once $X$ is traced out.

One can check that, if Bob sends register $X$ back to Alice (without copying it), $\QIC (\Pi, \mu)$ increases to $4 \lg |X|$ while $\QIC (\Pi^\prime, \mu)$ increases to $2 \lg |X|$ only. 
Moreover, if Bob first makes a copy of $X$ before sending it back, $\QIC (\Pi, \mu)$ increases to  $3 \lg |X|$ while $\QIC (\Pi^\prime, \mu)$ stays at $\lg |X|$.
By repeating this process for $r$ rounds, $\QIC (\Pi, \mu)$  increases to $(2r+ 1) \lg |X|$ while $\QIC (\Pi^\prime, \mu)$ stays at $\lg |X|$, and we can make these information costs as different as we like.
\end{example}
The examples above show that making safe copies might influence a lot the quantum information cost. However, we show that this operation can only decrease $\QIC$.

\begin{proposition}\label{prop:safeQIC}
	For any protocol $\Pi$ and any input distribution $\mu$ for $X, Y$, the safe version of $\Pi$, the protocol $\Pi^\prime$ defined above, satisfies
\begin{align}
	\QIC (\Pi^\prime, \mu) \leq \QIC (\Pi, \mu).
\end{align}
Moreover, if $\Pi$ is already a safe protocol, then we have equality.
\end{proposition}

\begin{proof}
	Before running protocol $\Pi$, let us first relabel the classical inputs $X$, $Y$ as $X^\prime$, $Y^\prime$, and then apply $U_X^{R_X \rightarrow R_X X}$ and $U_Y^{R_Y \rightarrow R_Y Y}$ on $R_X$, $R_Y$ in order to recreate coherent copies of the input in registers $X$, $Y$. The state at this point is then the same as in $\Pi^\prime$ before starting to apply the $U_i$'s (if we think of applying $U_Y$ on Bob's side before $U_1$ on Alice's side, which does not change the information cost), since that protocol is invariant under how the additional coherent copy of $X$ and $Y$ is created. If we then run $\Pi$ using the coherent copies in registers $X^\prime$, $Y^\prime$ as inputs, the state in each round is then the same as in $\Pi^\prime$. Notice that up to relabeling of the input registers and application of the isometries on $R_X$, $R_Y$, the protocol just defined is equivalent to $\Pi$, and hence it has the same information cost, with terms $I(R_X R_Y ; C_i | B_i)_{\rho_i} = I(R_X R_Y X Y ; C_i | B_i)_{\rho_i^\prime}$ in round $i$, in contrast to the information cost terms in $\Pi^\prime$, which are of the form $I(R_X R_Y ; C_i | Y B_i)_{\rho_i^\prime}$. The result follows since for each $i$, 
\begin{align}
I(R_X R_Y X Y ; C_i | B_i)_{\rho_i^\prime} & = I( Y  ; C_i | B_i)_{\rho_i^\prime} + I(R_X R_Y  ; C_i | Y B_i)_{\rho_i^\prime} + I( X ; C_i | R_X R_Y Y B_i)_{\rho_i^\prime} \\
				& \geq I(R_X R_Y  ; C_i | Y B_i)_{\rho_i^\prime},
\end{align}
and the terms $I( Y  ; C_i | B_i)_{\rho_i^\prime}$ and $I( X ; C_i | R_X R_Y Y B_i)_{\rho_i^\prime} = I( X ; C_i | A_i)_{\rho_i^\prime}$ vanishe whenever $\Pi$ is a safe protocol, holding throughout an unmodified copy of $X^\prime$ in $A_i$ and of $Y^\prime$ in $B_i$. The result follows.
\end{proof}

As a consequence, whenever we are interested in minimizing the quantum information cost, we may always consider such protocols that start by making a local copy of their inputs. This implies the following for the quantum information complexity of a function~:
\begin{corollary}\label{coro:QICisSafe}
	For any function $f$, any input distribution $\mu$,  and any $\epsilon>0$
	\begin{align}
		\QIC (f, \mu, \epsilon) = \inf_{\Pi'} \QIC (\Pi', \mu),
	\end{align}
	where the infimum is over the safe protocols $\Pi'$ computing $f$ with error $\epsilon$ w.r.t $\mu$.
\end{corollary}

Note that here, in contrast with \eqref{eq:defQIComplexity}, the minimum is over a smaller class of protocols.
In the sequel, unless otherwise specified, we only consider safe protocols.



\section{The Cost of Forgetting: a New Characterization of QIC}
\label{sec:cost-forget-new-charac-QIC}

In this section, we show that even though quantum protocols are reversible and thus can somehow forget information, there is
a quantum information cost associated in particular with forgetting classical information.
The fact, proven in the previous section, that unsafe protocols might have higher information cost than their safe counterpart can be seen as an example of this phenomenon for a party forgetting information about his own input. We focus here on safe protocols
and consider the cost of forgetting information learnt previously about the other party's input.
The remark at the end of Example~\ref{ex:forget} can be thought of as a simple, avoidable occurence of this phenomenon. We sill see later that in general for quantum protocols, it is not always possible to avoid this cost of forgetting information.

We introduce the Holevo Information Cost, defined as the amount of information the players have at the end of the protocol. We show that it is exactly characterized as the amount of information learnt minus the amount of information forgotten. This relation even holds at any intermediate stage of the protocol. We also consider how much Holevo information a party can obtain if he runs  (part of) his input in superposition.

Note that the information flow lemma, characterizing exactly the flow of quantum information in interactive protocols, can be seen as
 a fully quantum generalization of this result.

For protocols with classical inputs, we provide an alternative characterization of their quantum information cost that
does not require introducing a purification register.
More precisely, we show that at each round, $\QIC$ can be divided into two parts: the first one measures how much information is sent by one party to the other one; the second one counts how much information the party sending the message is forgetting about the other party's input. This additional term does not exist in classical communication  because players can always keep copies of all past messages, so they never forget information. But in quantum communication, cloning is in general impossible and players cannot always keep all the information they have received.

\subsection{Alternate Definitions of Information Costs for Protocols with Classical Inputs}
We first recall the notion of classical input information cost introduced by Kerenidis, Lauri\`ere, Le Gall and Rennela  in~\cite{KLLGR15,MR3497091}.
They also define an asymmetric version of quantum information cost.
They have the following definitions, in which we consider safe protocols and split Alice's local register in round $i$ as $X A_i$ and similarly as $Y B_i$ for Bob.

\begin{definition}
\label{charac:def:CIC}
	For a protocol $\Pi$ and an input distribution $\mu$, the \emph{classical input information cost of the messages from Alice to Bob} (resp. \emph{from Bob to Alice}) is defined as
\begin{align*}
						&\CIC_{A \rightarrow B} (\Pi, \mu) =  \sum_{i \geq 1, \ i \, odd} I (C_i ; X | Y B_i)
	\\
	\Big( \hbox{resp. } \quad &\CIC_{B \rightarrow A} (\Pi, \mu) =  \sum_{i \geq 1, \ i \, even} I (C_i ; Y | X A_i) \Big),
\end{align*}
and the \emph{quantum information cost of the messages from Alice to Bob} (resp. \emph{from Bob to Alice}) as
\begin{align*}
					&\QIC_{A \rightarrow B} (\Pi, \mu) =  \sum_{i \geq 1, \ i \, odd} I (C_i ; R_X R_Y | Y B_i)
	\\
	\Big( \hbox{resp. } \quad &\QIC_{B \rightarrow A} (\Pi, \mu) =  \sum_{i \geq 1, \ i \, even} I (C_i ; R_X R_Y | X A_i) \Big).
\end{align*}
\end{definition}

It follows from the data processing inequality that CIC is always at most QIC.
\begin{proposition}[\cite{KLLGR15,MR3497091}]
	For any protocol $\Pi$ and any input distribution $\mu$,
	\begin{equation}
		\CIC_{A \rightarrow B} (\Pi, \mu) \leq \QIC_{A \rightarrow B} (\Pi, \mu),
		\quad \CIC_{B \rightarrow A} (\Pi, \mu) \leq \QIC_{B \rightarrow A} (\Pi, \mu).
	\end{equation}
\end{proposition}

Note that $\QIC(\Pi, \mu) = \QIC_{A \rightarrow B} (\Pi, \mu) + \QIC_{B \rightarrow A} (\Pi, \mu)$, so we define similarly a symmetric version of classical input information cost of the protocol $\Pi$ as
\begin{align}
\CIC(\Pi, \mu) = \CIC_{A \rightarrow B} (\Pi, \mu) + \CIC_{B \rightarrow A} (\Pi, \mu).
\end{align}
We want to compare these two quantities, and in particular we find that they are related with a further notion of information cost, which we call the Holevo information cost. This quantity evaluates the Holevo information each party possesses at the end of the protocol about the other party's input, conditional on his own input.

\begin{definition}
\label{charac:def:HIC}
	For a protocol $\Pi$ and an input distribution $\mu$, the \emph{Holevo information cost from Alice to Bob} is defined as
\begin{align*}
	\HIC_{A \rightarrow B} (\Pi, \mu) = I (X ; B_{out} B^\prime | Y),
\end{align*}
and the \emph{Holevo information cost from Bob to Alice} as
\begin{align*}
	\HIC_{B \rightarrow A} (\Pi, \mu) = I (Y ; A_{out} A^\prime | X).
\end{align*}
We also define the \emph{(total) Holevo information cost} as $\HIC(\Pi, \mu) = \HIC_{A \rightarrow B} (\Pi, \mu) + \HIC_{B \rightarrow A} (\Pi, \mu)$.
\end{definition}
Note that similar considerations can be made in each round $i$ by considering the protocol $\Pi_i$ that runs $\Pi$ up to round $i$ and then stops (with an appropriate partition of the registers in round $i$, depending on whether $i$ is even or odd, and who holds $C_i$). For instance, in any odd round $i$, 
after reception by Bob of message $C_i$ from Alice,
the conditional Holevo information Bob has about Alice's input is: $I (X : B_i C_i | Y )$.
Such variants appeared, \textit{e.g.}, in Refs~\cite{JRS03, JN14}.

\subsection{Operational Interpretation of HIC in Terms of CIC and CRIC}
The quantity HIC corresponds to the information remaining at the end of the protocol. However, since in a quantum protocol it might be unavoidable to forget information along the way (because cloning is in general impossible), we cannot just count the information that was received: we should also quantify the amount of information that each player forgets. We introduce the following notion to take this phenomenon into account.
\begin{definition}
\label{charac:def:CRIC}
	For a protocol $\Pi$ and an input distribution $\mu$, the \emph{classical input reverse information cost of the messages from Bob back to Alice} (resp. \emph{from Alice to Bob}) is defined as
\begin{align*}
	&\CRIC_{A \leftarrow B} (\Pi, \mu) =  \sum_{i \geq 1, \ i \, even} I (C_i ; X | Y B_i)
	\\
	\Big( \hbox{resp. } \quad &\CRIC_{B \leftarrow A} (\Pi, \mu) =  \sum_{i \geq 1, \ i \, odd} I (C_i ; Y | X A_i) \Big).
\end{align*}
We also define the \emph{total classical input reverse information cost} of protocol $\Pi$ as 
\begin{equation*}
	\CRIC (\Pi, \mu) = \CRIC_{A \leftarrow B} (\Pi, \mu) + \CRIC_{B \leftarrow A} (\Pi, \mu).
\end{equation*}
\end{definition}
We soon make the above intuition more precise by providing an operational interpretation, but let us
first consider a simple example.
\begin{example}
Let $\mu$ be an input distribution with $X, Y$ distributed independently and uniformly on $n$ bits, and consider a protocol in which, 
after the second round, Alice has received a copy of Bob's input, $Y$. At this point, Alice copies 
the first $m$ out of the $n$ bits of $Y$, and sends back $Y$ to Bob. Then the term with $i=3$ in $\CRIC_{B \leftarrow A}$ will 
amount to the $(n-m)$ bits of information about $Y$ that Alice is forgetting.
\end{example}

We now suggest an operational interpretation of $\CIC$ and $\CRIC$. We can consider the following scenario. 
Let us fix a protocol $\Pi$. Consider a classical input state on registers $XY$ purified in registers $R_X R_Y$.
Alice is given her input $X$ as usual, but also the purification $R_Y$ of Bob's input.
 Bob is only given his input $Y$, and so only the register $R_X$ is held in some reference register unaccessible to the both parties. 
Alice is given the register $R_Y$ in order for her to be able to generate any 
state on $A_i B_i C_i$ in the protocol, for $i$ odd as well as $i$ even, and then transmit the message on $C_i$ to Bob, after giving 
him his side information $B_i $. We are interested in how much new 
information about $X$ this message $C_i$ contains, hence we are only 
putting $R_X$ in the referee's hand. More formally, suppose 
that we are interested in this information for round $i$. We then ask what 
is the asymptotic quantum communication cost for redistributing the $C_i$ register of 
this state from Alice to Bob if, apart from $C_i$, Alice holds the $A_i, X, R_Y$ 
registers and Bob holds the $B_i, Y$ registers. 
This is $I (C_i : R_X | B_i Y) = I(C_i : X |  Y B_i)$, for 
classical registers $X, Y$. Depending on whether $i$ is odd or 
even, this is the $i$th term in $\CIC_{A \rightarrow B}$ or in $\CRIC_{A \leftarrow B}$ of the protocol $\Pi$ (in the usual scenario where Alice does not have access to $R_Y$).
Remember that quantum communication in  state redistribution is symmetric 
under time-reversal~\cite{DY08, YD09}, so that the cost is the same if Bob decides to send back 
this message to Alice. Hence, not only does this scenario gives an operational interpretation 
to CIC as the amount of information about $X$ Alice is sending to Bob in odd 
rounds, but also to $\CRIC$ as the amount of information about $X$ Bob is 
forgetting by sending it back to Alice in even rounds.

This interpretation leads to the following formal result.
\begin{proposition}\label{thm:HicCicCric}
Let $\mu$ be a distribution and $\Pi$ be a safe protocol with classical inputs distributed according to $\mu$. Then
\begin{align*}
	\HIC_{A \rightarrow B} (\Pi, \mu) &= \CIC_{A \rightarrow B} (\Pi, \mu) - \CRIC_{A \leftarrow B} (\Pi, \mu), \\
	\HIC_{B \rightarrow A} (\Pi, \mu) &= \CIC_{B \rightarrow A} (\Pi, \mu) - \CRIC_{B \leftarrow A} (\Pi, \mu), \\
	\HIC(\Pi, \mu) &= \CIC(\Pi, \mu) - \CRIC(\Pi, \mu).
\end{align*}
\end{proposition}

\begin{proof}[Proof of Proposition~\ref{thm:HicCicCric}]

From the above operational interpretation of $\CIC$ and $\CRIC$, it is then intuitive that in any odd round $i$, 
after reception by Bob of message $C_i$ from Alice,
the conditional Holevo information $I (X^\prime : B_i C_i | Y^\prime )$ Bob has 
about Alice's input can be written as follows:
\begin{align}
\label{eq:holevo}
	I (X : B_i C_i | Y ) & = \sum_{j \, odd \, j \leq i }  I (C_j: X | Y B_j) - \sum_{j \, even \, j \leq i} I(C_j: X | Y B_j),
\end{align}
in which on the right hand side  the first sum corresponds to terms in $\CIC_{A \rightarrow B}$ 
and the second one to terms in $\CRIC_{A \leftarrow B}$.
Note that this equality follows from Corollary~\ref{lem:rdindeplb-coro}, direct consequence of the Information Flow Lemma, with classical extension registers $E_1 = X, E_2 = Y$ (classical copies of these registers), along with the fact that
for two classical copies $Y_1, Y_2$ of $Y$, $I(C_i ; X | Y_1 Y_2 B_i) = I (C_i ; X | Y B_i)$, $I (X; Y_1 B_i C_i | Y_2) = I (X ; B_i C_I | Y_2)$, and $I(X ; Y_1 | Y_2) = 0$.
If $r$ is odd, $I(X ; B_r C_r | Y) = I (X ; B_{out } B^\prime | Y)$ and the result follows. If $r$ is even, $I (X ; B_r C_r | Y) = I (X ; B_{out} B^\prime C_r | Y) = I (X ; B_{out} B^\prime  | Y) + I (X ;  C_r | Y B_r)$.
Similar statements hold for Alice, with the role of odd and even rounds interchanged. The statement follows.
\end{proof}

\subsection{Operational Interpretation of QIC in Terms of CIC and CRIC}

The introduction of the reference register $R$ in the definition of quantum information cost, which can be decomposed 
into $R = R_X, R_Y$ for classical inputs, is natural when discussing compression 
while keeping quantum correlations, and for general quantum inputs. But 
when discussing protocols implementing classical tasks it might appear somewhat artificial.
We now present an alternative characterization of quantum 
information cost on classical inputs that does not involve such purification registers and 
only mention the classical input registers, similar to the notion of 
classical input information cost (CIC) of Ref.~\cite{KLLGR15,MR3497091}.
We start by expanding the $i$th term in the quantum information 
cost. For odd $i$,
\begin{align}
	I(C_i : R_X R_Y | Y B_i)_{\rho_i^\prime} & = I(C_i : R_X  | Y B_i)_{\rho_i^\prime}
			+ I(C_i :  R_Y | R_X Y B_i)_{\rho_i^\prime}
\end{align}
(we could do similarly for even $i$ with the conditioning instead on $X^\prime A_i$).
The first term on the right hand side is the classical input information 
cost term $I (C_i ; R_X | Y B_i) = I (C_i ; X | Y B_i)$ in round $i$ and somehow quantifies the amount of information 
that message $C_i$ contains about $X$ for someone who already 
knows $Y$ and possesses $B_i$ as quantum side-information, while the second one does not 
immediately have such an intuitive interpretation. However, we 
can rewrite it as $I(C_i : R_Y | X A_i) = I (C_i : Y | X A_i)$ since $X A_i$ contain 
a purification of $\rho_i^{B_i C_i R_X R_Y Y}$. Notice that $X, Y$ are both classical 
in this term, which can now be informally interpreted as the amount of information that message $C_i$ 
contains about $Y$ for someone who already knows $X$ and possess $A_i$. But remember that it is Alice 
who generated message $C_i$, so in a classical protocol $A_i$ would contain a copy of $C_i$ and this term would always evaluate to $0$. 
However, quantum protocols are reversible, so it is somehow possible to forget information 
along the way. This term then corresponds, in a sense made precise by Proposition~\ref{thm:HicCicCric}, to the amount of information Alice is forgetting 
about $Y$ when transmitting $C_i$ ($\CRIC$).

This leads to the following result.
\begin{proposition}\label{thm:QicCicCric}
Let $\mu$ be a distribution and $\Pi$ be a safe protocol with classical inputs distributed according to $\mu$. Then
\begin{align*}
	\QIC_{A \rightarrow B} (\Pi, \mu) &= \CIC_{A \rightarrow B} (\Pi, \mu) + \CRIC_{B \leftarrow A} (\Pi, \mu), \\
	\QIC_{B \rightarrow A} (\Pi, \mu) &= \CIC_{B \rightarrow A} (\Pi, \mu) + \CRIC_{A \leftarrow B} (\Pi, \mu), \\
	\hbox{and }\qquad \QIC(\Pi, \mu) &= \CIC(\Pi, \mu) + \CRIC(\Pi, \mu).
\end{align*}
\end{proposition}

\subsection{QIC and CIC are Almost Equivalent}

We show that, even though the asymmetric versions of QIC and CIC 
can be very different as exhibited in Ref.~\cite{KLLGR15,MR3497091}, the symmetric 
versions can only be separated by at most a factor of two.
This can be understood intuitively by the fact that a protocol cannot forget more information than it transmits.
\begin{theorem}
	For any protocol $\Pi$ and any input distribution $\mu$, it holds that
	\begin{equation*}
		\CIC (\Pi, \mu) \leq \QIC (\Pi, \mu) \leq 2 \cdot \CIC (\Pi, \mu).
	\end{equation*}
	Hence for any function $f$, any input distribution $\mu$ and any error threshold $\epsilon$,
	\begin{equation*}
		\CIC (f, \mu, \epsilon) \leq \QIC (f, \mu, \epsilon) \leq 2 \cdot \CIC (f, \mu, \epsilon).
	\end{equation*}
\end{theorem}

It was already noticed in~Ref. \cite{KLLGR15,MR3497091}, that $\CIC(\Pi, \mu) \leq \QIC(\Pi, \mu)$. So to prove the above result, it is sufficient to show the following.
\begin{proposition}\label{prop:QICleqCIC}
	For any protocol $\Pi$ and any input distribution $\mu$, it holds that
\begin{align*}
	\QIC (\Pi, \mu) \leq 2 \cdot \CIC (\Pi, \mu).
\end{align*}
\end{proposition}
The proof relies on the characterization of the Holevo information cost given by Proposition~\ref{thm:HicCicCric}.

\begin{proof}
We have:
\begin{align}
	  \QIC (\Pi, \mu) & = \CIC (\Pi, \mu) + \CRIC(\Pi, \mu) \notag\\
		& \leq \CIC(\Pi, \mu) + \CRIC(\Pi, \mu) + \HIC(\Pi, \mu) \label{ineq:diffQICvsCIC} \\
		& = 2 \CIC(\Pi, \mu), \notag
\end{align}
where the inequality comes from the nonnegativity of Holevo information cost, that is
$\HIC (\Pi, \mu) \geq 0$, and the last equality holds by Proposition~\ref{thm:HicCicCric}.
\end{proof}

Since we believe that Proposition~\ref{prop:QICleqCIC} helps understanding $\QIC$ better  and might lead to new results involving this quantity, we provide an alternative proof sketch with a slightly different point of view. In particular, the symmetry of QIC with respect to a message being transmitted forward or backward is made evident, whereas the link between CIC and CRIC  under such a reversal of direction for message transmission is also highlighted.

\begin{proof}[Alternative Proof Sketch of~\thref{prop:QICleqCIC}]
Given a $r$-message protocol $\Pi$, let $\Pi^\prime$ be the protocol that runs $\Pi$ forward but does not discard $A^\prime$, $B^\prime$, and then, 
without making any copy of the output, runs $\Pi$ backward. Then, for any $k\in\{0,\dots, r-1\}$, the $(r+k)$th message in $\Pi'$ is identical 
to the $(r-k +1)$th message, except that the roles of the sender and receiver have been exchanged. 
Since the terms in $\QIC$ are symmetric under time-reversal, we have $\QIC_{A \rightarrow B} (\Pi^\prime, \mu) = \QIC_{B \rightarrow A} (\Pi^\prime, \mu) =   \QIC (\Pi, \mu)$. 
So the CIC for Alice and Bob in $\Pi'$ is respectively
\begin{align*}
	&\CIC_{A \rightarrow B} (\Pi^\prime, \mu) = \CIC_{A \rightarrow B} (\Pi, \mu) + \CRIC_{A \leftarrow B} (\Pi, \mu) = \QIC (\Pi, \mu)\\
	\hbox{and} \quad &\CIC_{B \rightarrow A} (\Pi^\prime, \mu) = \CIC_{B \rightarrow A} (\Pi, \mu) + \CRIC_{B \leftarrow A} (\Pi, \mu) = \QIC (\Pi, \mu),
\end{align*}
since the last $M$ messages in $\Pi^\prime$ consist of the $M$ messages of $\Pi$ run backward 
and thus the $\CIC$ of these messages in $\Pi'$ correspond to the $\CRIC$ of $\Pi$.
Thus, $\QIC (\Pi^\prime, \mu) = 2 \cdot \QIC (\Pi, \mu)$ and $\CIC (\Pi^\prime, \mu) = \QIC (\Pi, \mu)$.
 By~\eqref{eq:holevo} and the nonnegativity of Holevo information, $\CRIC_{A \leftarrow B} (\Pi, \mu)$ 
is at most $\CIC_{A \rightarrow B} (\Pi, \mu)$ and $\CRIC_{B \leftarrow A} (\Pi, \mu)$ 
is at most $\CIC_{B \rightarrow A} (\Pi, \mu)$, since it should not be possible to send back more information 
about the other party's input than what was received. This intuition also leads to the 
inequality $\QIC (\Pi, \mu) \leq 2 \cdot \CIC (\Pi, \mu)$.
\end{proof}

In~\thref{prop:QICleqCIC} we prove that QIC and CIC can be different by at most a factor of $2$. 
In fact, one can see that a necessary and sufficient condition to have $\QIC = \CIC$ is that $\CRIC(\Pi, \mu) = 0$, and then also $QIC = HIC$. Intuitively, this means that at each round the player who sends the message does not forget anything about what she has learnt in the previous rounds. Protocols with only a single message satisfy this property. Also, quantum simulation of classical protocols also satisfy this property; see Section~\ref{sec:quantum-simu-classical}.

At the other extreme, 
one can see that a sufficient condition to have $\QIC = 2 \cdot \CIC$ is that $\HIC (\Pi, \mu)= 0$, which only happens if the protocols completely uncompute any information about its input (apart possibly locally encoded information, or, as we will discuss later, ``phase'' or ``superposition'' information). Nevertheless, this bound should be almost achieved by memoryless protocols ({\it i.e.} protocols using only input registers together with a pure message register $C_i$, and no private working space registers $A_i$, $B_i$). Say the message register $C_r$ ends up with Bob, then $\QIC (\Pi, \mu) = \CIC (\Pi, \mu) - I (X;C_r|Y)$. However, players ould also forget information much later than they learn it, and so memoryless protocols are not the only type of protocols achieving this bound.

\subsection{Running Protocols on Superposition of Inputs}

In the previous section, we considered the amount of information a party learnt and forgot about the other party's classical input, when considering that he was also running on a classical input. However, in certain contexts, such as settings with privacy concerns~\cite{CvDNT99, MR2050849-K-02-qprivacy, JRS09, KLLGR15,SalvailSchaffnerSotakova-2015-quantifying}, other variants of the amount of information learnt by a party about the other party's classical input are natural to consider, like the one corresponding to  allowing that party to run on a quantum superposition of its intended input distribution. This makes for a quantum variant of the honest-but-curious classical paradigm, in which the party generates the correct ``distribution over messages'', but wishes to learn as much information as possible while doing so.

\subsubsection{Product Distributions}

With this in mind, we now define an alternative notion of quantum information cost for product distributions, and a corresponding decomposition of QIC, consistent with this idea. These definitions are ''superposed'' variants of the definitions in the previous sections.

\begin{definition}
\label{charac:def:SCIC-SCRIC-SHIC}
	For a protocol $\Pi$ and a product input distribution $\mu = \mu_X \otimes \mu_Y$, the \emph{superposed-classical input information cost of the messages from Alice to Bob} (resp. \emph{from Bob to Alice}) is defined as
\begin{align*}
						&\SCIC_{A \rightarrow B} (\Pi, \mu) =  \sum_{i \geq 1, \ i \, odd} I (C_i ; X | R_Y Y B_i)
	\\
	\Big( \hbox{resp. } \quad &\SCIC_{B \rightarrow A} (\Pi, \mu) =  \sum_{i \geq 1, \ i \, even} I (C_i ; Y | R_X X A_i) \Big),
\end{align*}
the \emph{superposed-classical input reverse information cost of the messages from Bob back to Alice} (resp. \emph{from Alice back to Bob}) is defined as
\begin{align*}
	&\SCRIC_{A \leftarrow B} (\Pi, \mu) =  \sum_{i \geq 1, \ i \, even} I (C_i ; X | R_Y Y B_i)
	\\
	\Big( \hbox{resp. } \quad &\SCRIC_{B \leftarrow A} (\Pi, \mu) =  \sum_{i \geq 1, \ i \, odd} I (C_i ; Y | R_X X A_i) \Big),
\end{align*}
the \emph{superposed-Holevo information cost from Alice to Bob} (resp. \emph{from Bob to Alice}) is defined as
\begin{align*}
	&\SHIC_{A \rightarrow B} (\Pi, \mu) = I (X ; R_Y Y B_{out} B^\prime )
\\
	\Big( \hbox{resp. } \quad &\SHIC_{B \rightarrow A} (\Pi, \mu) = I (Y ; R_X X A_{out} A^\prime  ) \Big).
\end{align*}
\end{definition}

Note that S-HIC is indeed the notion of information leakage considered by Ref.~\cite{JRS09} in their privacy trade-off for the index function on a uniform distribution.

We now link $\SCIC$ and $\SCRIC$ to $\QIC$ using the following remark.
For odd $i$,
\begin{align}
	I(C_i : R_X R_Y | Y B_i) & = I(C_i : R_Y  | Y B_i)
			+ I(C_i :  R_X | R_Y Y B_i)
\end{align}
(we could do similarly for even $i$ with the conditioning instead on $X A_i$).
The second term on the right hand side is the superposed-classical input information 
cost term $I (C_i ; R_X | R_Y Y B_i) = I (C_i ; X | R_Y Y B_i)$ in round $i$. For product distributions, it somehow quantifies the amount of information 
that message $C_i$ contains about $X$ for someone who runs the protocol with
the distribution corresponding to $Y$ in a superposition, and also possesses $B_i$ as quantum side-information. The first term does not 
immediately have such an intuitive interpretation. However, we 
can rewrite it as $I(C_i : R_Y  | Y B_i) = I(C_i : R_Y | R_X X A_i) = I (C_i : Y | R_X X A_i)$ since registers $R_X X A_i$ contains 
a purification of $\rho_i^{B_i C_i R_Y Y}$. It is then seen to be the superposed-classical input reverse information cost in round $i$, and hence corresponds to how much information Alice is forgetting about $Y$ if she runs the protocol with
the distribution corresponding to $X$ in a superposition, and also possesses $A_i$ as quantum side-information. It follows that $\QIC = \SCIC + \SCRIC$ (Note that this equality also formally holds for non-product distributions if we extend the definitions by using the corresponding CQMI terms).

The Information Flow Lemma can then be used to establish the link with SHIC, noting that for product distributions $I(X; Y R_Y)_{\rho_0} = I(X ; Y)_{\rho_i} = 0$, we obtain
\begin{align}
	\SHIC_{A \rightarrow B} (\Pi, \mu) & = \SCIC_{A \rightarrow B} (\Pi, \mu) - \SCRIC_{A \leftarrow B} (\Pi, \mu), \\
	\SHIC_{B \rightarrow A} (\Pi, \mu) & = \SCIC_{B \rightarrow A} (\Pi, \mu) - \SCRIC_{B \leftarrow A} (\Pi, \mu).
\end{align}

\subsubsection{General Distributions}

When considering non-product distributions, if Bob is to run his input in superposition, he should know (at least part of) Alice's input in order to ''break the correlations'' between their inputs, and allow him to generate the correct superposition consistent with Alice's input.
We consider how to do this for running only part of the input in superposition.
Notice that this encapsulates and extend both CIC, CRIC, HIC and their superposed variant at once.

Consider tensor product decomposition $X = X_1 \otimes X_2$ 
of Alice's input and $Y = Y_1 \otimes Y_2$ of Bob's input such that $X_1 Y_1$ 
and $X_2 Y_2$ are independent, i.e.~this gives a product 
decomposition $XY = X_1 Y_1 \otimes X_2 Y_2$. We can 
think of Bob running $Y_2$ in a quantum superposition, 
and so he also holds the purification $R_{X_2}$ of $X_2$ in order to 
generate the correct joint superposition consistent with Alice's input, 
while being given an actual classical input $Y_1$.
Alice is then also given a classical input in $X_1$ (and we can think of $X_2$ either as a classical input whose classical copy or purification is initially held by Bob, or as a superposition over classical inputs jointly held by Alice and Bob).
The corresponding hybrid information costs are defined as follows, with similar definitions  for Alice.

\begin{definition}
\label{charac:def:HCIC-HCRIC-HHIC}
	For a protocol $\Pi$ and an arbitrary decomposition $X = X_1 \otimes X_2$, $Y = Y_1 \otimes Y_2$ of the input space, and arbitrary distributions $\mu_1$ on $X_1 Y_1$ and $\mu_2$ on $X_2 Y_2$,  when running $\Pi$ on input distribution $\mu_1 \otimes \mu_2$, the \emph{hybrid-classical input information cost of the messages from Alice to Bob} (resp. \emph{from Bob to Alice}) is defined as
\begin{align*}
						&\HCIC_{A \rightarrow B} (\Pi, \mu_1, \mu_2) =  \sum_{i \geq 1, \ i \, odd} I (C_i ; X_1 | R_{X_2} R_{Y_2} Y_1 Y_2 B_i)
	\\
		& \quad\quad\quad\quad\quad\quad\quad \quad\quad =  \sum_{i \geq 1, \ i \, odd} I (C_i ; X_1 | X_2 R_{Y_2} Y_1 Y_2 B_i)
	\\
	\Big( \hbox{resp. } \quad &\HCIC_{B \rightarrow A} (\Pi, \mu_1, \mu_2) =  \sum_{i \geq 1, \ i \, even} I (C_i ; Y_1 | R_{Y_2} R_{X_2} X_1 X_2 A_i)
	\\
		& \quad\quad\quad\quad\quad\quad \quad\quad\quad =  \sum_{i \geq 1, \ i \, even} I (C_i ; Y_1 | Y_2 R_{X_2} X_1 X_2 A_i) \Big),
\end{align*}
the \emph{hybrid-classical input reverse information cost of the messages from Bob back to Alice} (resp. \emph{from Alice to Bob}) is defined as
\begin{align*}
	&\HCRIC_{A \leftarrow B} (\Pi, \mu_1, \mu_2) =  \sum_{i \geq 1, \ i \, even} I (C_i ; X_1 | R_{X_2} R_{Y_2} Y_1 Y_2 B_i)
	\\
		& \quad\quad\quad\quad\quad\quad\quad \quad\quad =  \sum_{i \geq 1, \ i \, even} I (C_i ; X_1 | X_2 R_{Y_2} Y_1 Y_2 B_i)
	\\
	\Big( \hbox{resp. } \quad &\HCRIC_{B \leftarrow A} (\Pi, \mu_1, \mu_2) =  \sum_{i \geq 1, \ i \, odd} I (C_i ; Y_1 | R_{Y_2} R_{X_2} X_1 X_2 A_i)
	\\
		& \quad\quad\quad\quad\quad\quad \quad\quad\quad =  \sum_{i \geq 1, \ i \, odd} I (C_i ; Y_1 | Y_2 R_{X_2} X_1 X_2 A_i) \Big),
\end{align*}
the \emph{hybrid-Holevo information cost from Alice to Bob} (resp. \emph{from Bob to Alice}) is defined as
\begin{align*}
	&\HHIC_{A \rightarrow B} (\Pi, \mu_1, \mu_2) = I (X_1^\prime ; R_{Y_2} Y_2^\prime B_{out} B^\prime | Y_1^\prime X_2^\prime )
\\
	\Big( \hbox{resp. } \quad &\HHIC_{B \rightarrow A} (\Pi, \mu_1, \mu_2) = I (Y_1^\prime ; R_{X_2} X_2^\prime A_{out} A^\prime | X_1^\prime Y_2^\prime  ) \Big).
\end{align*}
\end{definition}

Note that by the Information Flow Lemma and the fact that $X_1$ and $X_2$ (resp., $Y_1$ and $Y_2$) are independent, we get that
\begin{align}
	\HHIC_{A \rightarrow B} (\Pi, \mu_1, \mu_2) & = \HCIC_{A \rightarrow B} (\Pi, \mu_1, \mu_2) - \HCRIC_{A \leftarrow B} (\Pi, \mu_1, \mu_2), \\
	\HHIC_{B \rightarrow A} (\Pi, \mu_1, \mu_2) & = \HCIC_{B \rightarrow A} (\Pi, \mu_1, \mu_2) - \HCRIC_{B \leftarrow A} (\Pi, \mu_1, \mu_2).
\end{align}

We then say that Alice does not forget information if the $\HCRIC$ from Bob to Alice is $0$ for any decomposition of the inputs. More formally, we introduce the following definition.

\begin{definition}
\label{def:protnoforget}
	Given a protocol $\Pi$, we say that \emph{Alice (resp. Bob) does not forget information} in $\Pi$ if for any decomposition $X = X_1 \otimes X_2$, $Y = Y_1 \otimes Y_2$ of the input space, and any distributions $\mu_1$ on $X_1 Y_1$ and $\mu_2$ on $X_2 Y_2$, it holds that
\begin{align*}
	&\HCRIC_{A \leftarrow B} (\Pi, \mu_1, \mu_2) = 0 
\\
	\Big( \hbox{resp. } \quad &\HCRIC_{B \leftarrow A} (\Pi, \mu_1, \mu_2) =  0 \Big).
\end{align*}
We say that protocol $\Pi$ \emph{does not forget information} if both Alice and Bob do not forget information in $\Pi$.
\end{definition}

\begin{remark}
\label{charac:rem:equivNF}
	In particular, if a protocol $\Pi$ does not forget information, for any input distribution $\mu$, $\CRIC(\Pi, \mu) = 0$, and $\QIC(\Pi, \mu) = \HIC(\Pi, \mu) = \CIC(\Pi, \mu)$.
\end{remark}



\section{Forgetting Information in Classical Protocols}\label{sec:clforget}

We considered quantum protocols forgetting classical messages by viewing such messages as part of a quantum register, 
on which we could apply a reversible quantum operation in order to generate the subsequent message. In the same way, we can consider  a reversible classical computation paradigm where classical protocols can forget information. We will show that such an ability does not provide any advantage over protocols in the standard classical information complexity paradigm: for any protocol that can forget information, there exists a protocol that does not forget information with the same input-output behavior, the same amount of communication, and information cost at most that of the protocol that can forget information. In this section, all the protocols we consider are classical.

\subsection{Extending the Classical Setting : a New Characterization of IC}
Let us begin by deriving some alternative characterization of classical information complexity that will enable easier comparison to the quantum setting. Let us first state some definitions. In the sequel, unless otherwise specified, we denote $S_A, S_B$, and $R_{AB}$ the random variables corresponding respectively to the private coins of Alice, of Bob, and the public randomness. 
\begin{definition}
	A (standard) $r$-round classical protocol $\pi$ is defined by the sequence of its message functions such that~: for all odd $1 \leq i \leq r$,
	$m_i$ is a function of $(x, s_A, r_{AB}, m_{<i})$,
	and for all even $2 \leq i \leq r$,
	$m_i$ is a function of $(y, s_B, r_{AB}, m_{<i})$.
\end{definition}
The randomness of a protocol is contained on the one hand in the inputs $(X,Y)$ and on the other hand in the random coins $(S_A,S_B,R_{AB})$.
\begin{definition}
\label{forgetclassical:def:defIC}
	The (standard) information cost of a protocol $\pi$ with transcript $\Pi = M_1 \cdots M_r$ on input distribution $\mu$ is~:
$$
	IC (\Pi, \mu) = IC_{A \rightarrow B} (\Pi, \mu) + IC_{B \rightarrow A} (\Pi, \mu),
$$
where $IC_{A \rightarrow B} (\Pi, \mu) = I (X ; \Pi | R_{AB} Y)$, and $IC_{B \rightarrow A} (\Pi, \mu) = I(Y ; \Pi | R_{AB} Y)$ are respectively the information costs from Alice to Bob and from Bob to Alice, and $\Pi$ is the sequence of messages.
\end{definition}

We generalize the above definitions to the case where there is an additional random variable correlated with the input.
\begin{definition}
	Given a random variable $U$ with distribution $\mu$, we say that a joint random variable $UV$ is an extension of $U$, or that $V$ extends $U$, if the marginal of $UV$ on $U$ has distribution $\mu$. 

\end{definition}
Moreover, we say that $V$ is a copy of $U$ if $\PP(U=V)=1$.
\begin{lemma}
For any protocol $\Pi$, any input distribution $\mu$ on $XY$ and any extension  $XY X^\prime Y^\prime D$ of $XY$, where $X^\prime Y^\prime$ are copies of $XY$, it holds that:
\begin{align}
	IC (\Pi, \mu)
	&= \sum_{i\,:\, \hbox{odd}} I(X^\prime Y^\prime D ; M_i | R_{AB} S_B Y M_{<i}) +
\sum_{i\,:\, \hbox{even}} I(X^\prime Y^\prime D ; M_i | R_{AB} S_A X M_{<i}) \nonumber
	\\
	&= \sum_{i} \Big( I(X^\prime Y^\prime D ; M_i | R_{AB} S_B Y M_{<i})
+  I(X^\prime Y^\prime D ; M_i | R_{AB} S_A X M_{<i}) \Big) \label{eq:ICRev}
	\\
	&= I(X^\prime Y^\prime D ; \Pi | R_{AB} S_B Y) + I(X^\prime Y^\prime D ; \Pi | R_{AB} S_A X) \nonumber
	\end{align}
\end{lemma}
\begin{proof}
	For the first equality, let us consider the right-hand side. In any odd round $i$, we have~:
	\begin{align*}
	&I (X^\prime Y^\prime D; M_i | R_{AB} S_B Y M_{<i} )
	\\
	&=
	I (Y^\prime ; M_i | R_{AB} S_B Y M_{<i} )
	+
	I (X^\prime; M_i | Y^\prime R_{AB} S_B Y M_{<i} )
	+
	I (D; M_i | X^\prime Y^\prime R_{AB} S_B Y M_{<i} )
	\\
	&=
	I(X; M_i | R_{AB} S_B Y M_{<i})
	\\
	&=
	I(X; M_i | R_{AB} Y M_{<i}),
	\end{align*}
	where we used the following facts.
	Firstly, $I (Y^\prime ; M_i | R_{AB} S_B Y M_{<i} )=0$, since all the quantities are classical and $Y$ appears in the conditioning.
	Secondly, $I (D; M_i | XR_{AB} S_B Y M_{<i} ) = 0$; indeed, by the Markov property of $\Pi$, conditioned on $XYR_{AB}S_B M_{<i}$, $M_i$ is independent of $D$. 
Finally, conditioned on either of $XY R_{AB} M_{<i}$ or $Y R_{AB} M_{<i}$, the message $M_i$ generated by Alice is independent of $S_B$.
Similarly, in any even round $i$, we have~:
$$
	I (X^\prime Y^\prime D; M_i | X R_{AB} S_A M_{<i} ) = I(Y; M_i | X R_{AB} M_{<i}).
$$	
Summing over rounds and using the chain rule of conditional mutual information and Definition~\ref{forgetclassical:def:defIC} yields the first equality.

For the second equality, note that for any odd $i$
	\begin{align*}
	&I (X^\prime Y^\prime D; M_i | X R_{AB} S_A M_{<i} )
	\\
	&=
	I (X^\prime ; M_i | X R_{AB} S_A M_{<i} )
	+
	I (Y^\prime; M_i | X^\prime X R_{AB} S_A M_{<i} )
	+
	I (D; M_i | Y^\prime X^\prime X R_{AB} S_A M_{<i} )
	\\
	&=
	I(Y; M_i | X R_{AB} S_A M_{<i})
	\\
	&= 0,
	\end{align*}
in which the last equality follows since $M_i$ is a deterministic function of $X R_{AB} S_A M_{<i}$.
	Similarly, in any even round $i$, we have~:
	$$
	I (X^\prime Y^\prime D; M_i | R_{AB} S_B Y M_{<i} ) = 0.
	$$
	The last equality holds by the chain rule for conditional mutual information.
\end{proof}

The form~(\ref{eq:ICRev}) has a natural interpretation, which we will adopt to define information cost in the reversible classical computation paradigm that we study in the next subsection: it quantifies how much information message $M_i$ in round $i$ contains about  any extension of the input, conditional on the information already known at the receiver's side for one term, and on the sender's side for the other term.
Since communication in protocols in the reversible classical computation paradigm should be symmetric under time reversal, this will be the natural extension of IC that we will study in that paradigm.

\subsection{Reversible Classical Protocols}\label{subsec-revCP}
For notational simplicity, given two registers $I$ and $O$, we will denote $\circuit^{I \rightarrow O}$ a reversible circuit taking  $I$ as input and outputting in $O$.
\begin{definition}
\label{forgetclassical:def:reversibleproto}
A reversible $r$-round classical protocol taking $X,Y$ as inputs, with private randomness $S_A,S_B$ and public randomness $R_{AB}^A,R_{AB}^B$ (each player has a copy of the public randomness), and outputting in $A_{out}B_{out}$, is defined by a sequence of reversible circuits~:
	$\circuit_1^{X S_A R_{AB}^A \rightarrow A_1 M_1}$,
	$\circuit_2^{Y S_B R_{AB}^B M_1 \rightarrow B_2 M_2}$,
	$\circuit_3^{A_1 M_2 \rightarrow A_3 M_3}$,
	$\circuit_4^{B_2 M_3 \rightarrow B_4 M_4}$,
	$\cdots$,
	$\circuit_r^{A_{r-2} M_{r-1} \rightarrow A^\prime A_{out} M_r}$,
	$\circuit_{r+1}^{B_{r-1} M_r \rightarrow B^\prime B_{out} }$. The circuits are applied in turn by each player to produce the messages $M_i$ (plus some private data $A_i$ or $B_i$ respectively for $i$ odd and $i$ even).
\end{definition}
For simplicity, we relabel $A_i = A_{i-1}$ for odd $i$ and $B_i = B_{i-1}$ for even $i$.

As in the quantum case, we will often require that the players make a copy of their inputs at the outset of the protocol, and run the protocol on these copies while leaving their original input registers unmodified.
\begin{definition}
A reversible classical protocol is said to be safe if it leaves the input registers unmodified. 
The safe version of an arbitrary reversible protocol is one in which the players start by making safe copies of their inputs, and then run the protocol on the copies.
\end{definition}
We now define a notion of information cost associated to a reversible protocol.
\begin{definition}
\label{forgetclassical:def:defRIC}
Let $\Pi$ be a reversible protocol as per Definition~\ref{forgetclassical:def:reversibleproto}, let $\mu$ be an input distribution on $XY$, and let $XY X^\prime Y^\prime D$ be any extension of $XY$, with $X^\prime Y^\prime$ being copies of $XY$. The reversible information cost of $\Pi$ on input distribution $\mu$ is defined as~:
\begin{align}
\label{forgetclassical:eq:defRIC}
	\RIC(\Pi, \mu)
	=
	\sum_{i} I (X^\prime Y^\prime D; M_i | B_{i})
	+ \sum_{i} I (X^\prime Y^\prime D; M_i | A_{i}).
\end{align}
\end{definition}
Note that the sum is over all rounds for both terms of the right-hand side.
We first make sure that the above is well-defined, and does not depend on the choice of extension in $D$. Also, as in the quantum setting, we show that making safe copies does not increase RIC.
\begin{lemma}
\label{forgetclassical:lem:safeclassical}
	For any reversible protocol $\Pi$ and input distribution $\mu$, measuring information about $X^\prime Y^\prime$ is sufficient: for any extension $XY X^\prime Y^\prime D$ as above, it holds that
\begin{align}\label{eq:RICsumExt}
	\RIC(\Pi, \mu)
	=
	\sum_{i} I (X^\prime Y^\prime ; M_i | B_{i})
	+ \sum_{i} I (X^\prime Y^\prime ; M_i | A_{i}).
\end{align}
Moreover, denoting $\Pi^\prime$ the safe version of $\Pi$, it holds that
\begin{align}\label{eq:RICsafeIneq}
	\RIC(\Pi^\prime, \mu)
	\leq
	\RIC(\Pi, \mu).
\end{align}
\end{lemma}
\begin{proof}
By the Markov property of $\Pi$, it holds that, conditional on $X^\prime Y^\prime B_i$ or $X^\prime Y^\prime A_i$, $M_i$ and $D$ are independent. The equality~\eqref{eq:RICsumExt} follows.

For the safe protocol $\Pi^\prime$, let $X^{\prime \prime} Y^{\prime \prime}$ be safe copies made at the outset to take as input to $\Pi$. Then Alice holds $X$ throughout together with $A_i$, and Bob holds $Y$ throughout together with $B_i$. It holds that
$$
I (X^\prime Y^\prime D ; M_i | Y B_i) \leq I(X Y X^\prime Y^\prime D ; M_i |  B_i)
$$
and
$$
I (X^\prime Y^\prime D ; M_i | X A_i) \leq I(X Y X^\prime Y^\prime D ; M_i |  A_i).
$$
Then, relabeling inputs $XY$ to $\Pi$ as $X^{\prime \prime} Y^{\prime \prime}$, thinking of $X Y X^\prime Y^\prime D$ as an extension of input $X^{\prime \prime} Y^{\prime \prime}$, and applying~\eqref{eq:RICsumExt} to $\Pi^\prime$ and $\Pi$ leads to~\eqref{eq:RICsafeIneq}.
\end{proof}

We thus consider only safe protocols, denote Alice's and Bob's local memory registers as $X A_i$, $Y B_i$, respectively,
and use the following characterization of information cost for these protocols~: 
$$
	\RIC(\Pi, \mu) =	\sum_{i} I (X ; M_i | Y B_{i-1})
	+ \sum_{i} I (Y ; M_i | X A_{i-1}).
$$
For standard IC, we can restrict the sum measuring information about $X$ to odd messages, and the sum measuring information about $Y$ to even messages. The additional terms here are similar to those in the quantum setting and correspond to the cost of forgetting information in a classical reversible paradigm. We want to show that forgetting is in fact useless here. The following lemma will be useful towards this goal.
\begin{lemma}
\label{lem:reversible-classical-subadd}
The reversible information cost is subadditive~: for any two protocols $\Pi_1, \Pi_2$, and any joint input distribution $\mu_{12}$ on  $X_1Y_1X_2Y_2\,$,
$$
	\RIC \big(\Pi_1 \otimes \Pi_2, \mu_{12} \big) \leq \RIC (\Pi_1, \mu_1) + \RIC (\Pi_2, \mu_2),
$$
with $\mu_1$ the marginal of $\mu_{12}$ on $X_1 Y_1$, and $\mu_2$ the marginal of $\mu_{12}$ on $X_2 Y_2$.
\end{lemma}
\begin{proof}
	Consider an odd round $i$ (Bob is the receiver). The $i$-th term in $\RIC$ of $\Pi_1 \otimes \Pi_2$ on input $X_1 X_2 Y_1 Y_2$ with extension $X_1 X_2 Y_1 Y_2 D$ is~:
	\begin{align}
		&I (X_1 X_2 Y_1 Y_2 D; M_{1,i} M_{2,i} | B_{1,i-1}B_{2,i-1})
		\notag
		\\
		&= I (X_1 X_2 Y_1 Y_2 D; M_{1,i}  | B_{1,i-1}B_{2,i-1})
		+ I (X_1 X_2 Y_1 Y_2 D; M_{2,i} | M_{1,i} B_{1,i-1}B_{2,i-1}).
		\label{forgetclassical:eq:termRICiSubadd}
	\end{align}
	The $i$-th term in $\RIC$ of $\Pi_1$  on  inputs $X_1 Y_1$ with extension $X_1 Y_1 X_2 Y_2 D$ is~:
	\begin{align*}
		&I (X_1 Y_1 X_2 Y_2 D; M_{1,i} | B_{1,i-1})
		\\
		= \,& I (X_1 Y_1 X_2 Y_2 D A_{2,i-1} M_{2,i-1} B_{2,i-1}; M_{1,i} | B_{1,i-1})
		\\
		= \,& I (B_{2,i-1}; M_{1,i} | B_{1,i-1})
		+ I (X_1 Y_1 X_2 Y_2 D  ; M_{1,i} | B_{1,i-1} B_{2,i-1} )
		\\
		&\qquad + I (A_{2,i-1} M_{2,i-1} ; M_{1,i-1} | X_1 Y_1 X_2 Y_2 D  B_{1,i-1} B_{2,i-1}) 
		\\
		\geq \,& I (X_1 Y_1 X_2 Y_2 D  ; M_{1,i} | B_{1,i-1} B_{2,i-1} )
	\end{align*}
	which is the first term in~\eqref{forgetclassical:eq:termRICiSubadd}. Above, the first equality is by first appending uncorrelated registers $S_A S_B R_{AB}^A R_{AB}^B$, and then by invariance of conditional mutual information (CMI) under local reversible processing. The second equality is by the chain rule, and the inequality holds by non-negativity of the CMI.
	 
	To obtain the second term in~\eqref{forgetclassical:eq:termRICiSubadd}, let us rewrite the $i$-th term in $\RIC$ of $\Pi_2$  on  inputs $X_2 Y_2$ with extension $X_1 Y_1 X_2 Y_2 D$ as~:
	\begin{align*}
		&I (X_1 Y_1 X_2 Y_2 D; M_{2,i} | B_{2,i-1})
		\\
		= \,& I (X_1 Y_1 X_2 Y_2 D A_{1,i} M_{1,i} B_{1,i}; M_{2,i} | B_{2,i-1})
		\\
		= \,& I (M_{1,i} B_{1,i}; M_{2,i} | B_{2,i-1})
		+ I (X_1 Y_1 X_2 Y_2 D ; M_{2,i} | M_{1,i} B_{1,i} B_{2,i-1})
		\\
		&\qquad + I ( A_{1,i} ; M_{2,i} | X_1 Y_1 X_2 Y_2 D M_{1,i} B_{1,i} B_{2,i-1})
		\\
		\geq \,& I (X_1 Y_1 X_2 Y_2 D ; M_{2,i} | M_{1,i} B_{1,i-1} B_{2,i-1}),
	\end{align*}
	with similar arguments as above (and the fact that, since Bob is the receiver, $B_{1,i}=B_{1,i-1}$).
	 We similarly control
\begin{align}
	 I(X_1 X_2 Y_1 Y_2 D; M_{1,i} M_{2,i} & | A_{1,i} A_{2,i})
		\notag
		\\
		&\leq I (X_1 X_2 Y_1 Y_2 D; M_{1,i}  | A_{1,i})
		+ I (X_1 X_2 Y_1 Y_2 D; M_{2,i} | A_{2,i}).
\end{align}
	For any even round $i$, we obtain similar relationships between the various RIC terms. Summing over rounds yields the conclusion.
\end{proof}

\begin{theorem}
\label{thm:simu-classical-NF}
It is possible to simulate any reversible protocol $\Pi$ by a (standard) protocol $\Pi^\prime$ that does not forget information without increasing the information or the communication costs.
\end{theorem}
\begin{proof}
	Let $\Pi$ be a reversible protocol. 
	We assume, without blow up in the information and the communication costs that the protocol makes local copies of the inputs (see Lemma~\ref{forgetclassical:lem:safeclassical}). 
	We define $\Pi^\prime$ as follows: the players run $\Pi$, but with each party makes a copy of the message in each round and not further acts on that copy.
Then, at round $i$, we can view the action of the protocol as the combined action of two one-round protocols : $\Pi_1^i$, which is a reversible protocol implementing the new message by taking the local registers of the reversible protocol as input, and $\Pi_2^i$, which contains the previous messages as side information and does not send any message. Then, we use the subadditivity of RIC (see Lemma~\ref{lem:reversible-classical-subadd}) on these two protocols. Summing over the rounds, we obtain the desired simulation, since these yield the corresponding RIC of the reversible protocol and its standard version.
\end{proof}



\section{Disjointness: Speed-up for Quantum Protocols needs Forgetting Information}
\label{sec:disj}

In light of what we saw for classical protocols that forget information,  the phenomenon of forgetting information in a quantum protocol  might appear useless, or even costly, at first sight. A legitimate question is: \emph{given any safe quantum protocol implementing a classical task, potentially forgetting information, is there a protocol
 that does not forget information and accomplishes the same task at a similar information cost?}
 We give a strong negative answer to this question in the case of the Disjointness problem, showing that the ability to forget information is a necessary quantum feature to obtain any speed-up for computing disjointness.

Recently, the notion of $\QIC$ was used by Braverman et al.~\cite{BGKMT15} to prove an 
optimal lower bound, up to logarithmic terms, on the bounded-round quantum communication 
complexity of the disjointness function for $n$-bit inputs, defined as: for all $x,y \in \ZO^n$,
\begin{align*}
\DISJ_n(x,y) = \neg \left(\OR_{i \in [n]} (x_i \, \AND \, y_i)\right). 
\end{align*}
The authors proved that, for a given number $r$ of rounds of communication, the quantum communication complexity is $\QCC^r (\DISJ_n) \in \tilde{\Omega} (\frac{n}{r} + r)$. 
We adapt their proof to show that, if we only allow quantum protocols that  do not forget information, 
then the round dependence disappear and we completely lose the quadratic quantum speed-up for computing disjointness.  
This establishes the fact that, in contrast to the case for classical information cost, the ability to forget information 
is a necessary feature of quantum protocols.

The high-level idea of the proof in Ref.~\cite{BGKMT15} can be described as follows. The $\QIC$ of any protocol solving $\DISJ_n$ 
is lower bounded by $n$ times the QIC of a protocol solving $\AND$, in which the information is measured with respect to 
any distribution having zero mass on $(1,1)$ input. The lower bound on the bounded-round quantum communication  for disjointness then follows from the fact that any protocol solving $\AND$ must have $\QIC$ 
at least $\tilde{\Omega} (\frac{1}{r})$ on such distributions.
This lower bound for $\AND$ is in turn proven by reducing back to disjointness, 
for which they prove that $\QIC (\DISJ_m) \in \Omega (\sqrt{m})$ (for any $m \in \NN$), and then constructing 
a low-information protocol for disjointness by applying coordinate-wise 
some low-information protocol for $\AND$. 
The authors were interested in the regime $m \in \tilde{\Theta} (r^2)$.
By appropriately subsampling, we can ensure that the QIC of the 
constructed protocol is close to $m$ times that of the AND protocol 
on distributions with zero-mass on $(1,1)$ inputs. The remaining ingredient is a bound on the 
continuity of QIC in the input distribution.

In fact, this continuity argument is the only place where round complexity comes into play. 
For the AND function, it states that a $r$-round protocol $\Pi$ run on 
an input distribution with mass $w$ on $(1,1 )$
input has QIC which is $(r \cdot H(w))$-close to the one of $\Pi$ run 
on some input distribution with $0$-mass on $(1,1 )$-input.
Note that this factor of $r$ is not present for classical information cost 
(unless we allow for forgetting information, as in Section~\ref{sec:clforget}, in which case it is also there in general) and, 
at an intuitive level, it can be thought of as arising from the possibility of quantum protocols transmitting $r$ 
times the same information about the $(1,1)$ input. In particular, it is not there for quantum protocols that do 
not forget information, and this is the reason why we can lift the proof of Ref.~\cite{BGKMT15} to a linear lower bound for such protocols. 
We formalize this intuition below.

\begin{definition}
	We denote $\T^{r, NF} (f, \epsilon)$ the set of $r$-round protocols that solve $f$ with error at most $\epsilon$ and do not forget information as per Definition~\ref{def:protnoforget}.
\end{definition}
\begin{definition}
	We denote $\QCC^{r, NF} (f, \epsilon)$ (resp. $\QIC^{r, NF} (f, \epsilon)$) the minimal communication (resp. information) cost achieved by a $r$-round quantum protocol solving $f$ with error at most $\epsilon$, and without forgetting information -- that is:
	$$
		\QCC^{r, NF} (f, \epsilon) = \min_{\Pi \in  \T^{r, NF} (f, \epsilon)} \QCC(\Pi),
		\qquad
		\QIC^{r, NF} (f, \epsilon) =  \inf_{\Pi \in  \T^{r, NF} (f, \epsilon)} \max_\mu \QIC(\Pi, \mu).
	$$
\end{definition}
We prove that any protocol solving $\DISJ_n$ without forgetting information must have communication $\Omega(n)$.
\begin{theorem}\label{thm:disj-NF-n}
	$$
	\QCC^{r, NF} (\DISJ_n, 1/3) \in \Omega (n).
	$$
\end{theorem}
First, we can obtain the following result by going over the proof of the corresponding result (Lemma~4.20) in Ref.~\cite{BGKMT15} and restricting our attention to protocols that do not forget information. 
The proof, given for completeness, is deferred to the Appendix (see Appendix~\ref{subsec:proof-lem-QCC-DISJ-QIC-AND}). We require an additional definition.
\begin{definition}

	We denote $\QIC_0^{r, NF} (\AND, \epsilon)$ the minimal information cost on input distributions with no support on $(1,1)$ inputs achieved by a $r$-round quantum protocol solving $AND$ with error at most $\epsilon$, and without forgetting information -- that is:
	$$
		\QIC_0^{r, NF} (\AND, \epsilon) =  \inf_{\Pi \in  \T^{r, NF} (\AND, \epsilon)} \max_{\mu_0} \QIC(\Pi, \mu_0),
	$$
in which the maximum is taken over all input distribution satisfying $\mu_0 (1, 1) = 0$.
\end{definition}

\begin{lemma}\label{lem:QCC-DISJ-QIC-AND}
	$\QCC^{r, NF} (\DISJ_n, 1/3) \geq n \cdot \QIC_0^{r, NF} (\AND, 1/3)$.
\end{lemma}

Furthermore, we adapt the proof of Corollary 4.9 in Ref.~\cite{BGKMT15} for protocols not forgetting information and obtain the following result.
The proof is deferred to the Appendix (see Appendix~\ref{subsec:proof-lem-QICmu0}).
\begin{lemma}
		\label{lem:QICmu0}
		Suppose we have a protocol $\Pi$ for $\AND$ which does not forget information. Then, for any input distribution $\mu$ not concentrated on $(1,1)$,
		$$
			\QIC(\Pi, \mu) \leq \QIC(\Pi, \mu_0) + H(w)
		$$
		(independently of the number of rounds in $\Pi$),
		where $w = \mu(1,1) \leq 1/2$, $\mu_0(1,1) = 0$, $\mu_0(x,y) = \frac{1}{1-w} \mu(x,y)$ for $(x,y) \neq (1,1)$.
\end{lemma}
	
A protocol that does not forget information can be boosted without forgetting information or increasing the number of round, similarly to Lemma~4.15 of Ref.~\cite{BGKMT15}.

\begin{lemma}\label{lem:boosterror}
For any function $f$, any bound on the number of round $r$ and any error parameter $\epsilon>0$, the following holds:
\begin{align}
\QIC^{r, NF} (f, \epsilon) \leq O (\lg 1/\epsilon) \QIC^{r, NF} (f, 1/3).
\end{align}
\end{lemma}

We make use of the following lower and upper bounds proven in Ref.~\cite{BGKMT15} (the upper bound follows from the proof of  their Lemma~6.1) on the QIC of computing $\DISJ_m$ for some parameter $m \in \mathbb{N}$.

\begin{lemma}\label{lem:LB-QIC-DISJ}
	$\QIC (\DISJ_m, 1/3) \in \Omega (\sqrt{m})$.
\end{lemma}

\begin{lemma}\label{lem:QIC-DISJ-AND}
For any $m$, any protocol $\Pi_A$ computing $\AND$ with error $1/m^2$, and any $w \in  O(\lg^4(m) / m)$,
$$
	\QIC (\DISJ_m, 2/m) \leq m  \cdot  \max_{\mu_w} \QIC (\Pi_A, \mu_w) + o (\sqrt{m}),
$$
in which $\mu_w$ ranges over all distributions with weight at most $w$ on the $(1,1 )$-input.
\end{lemma}

Optimizing over protocols $\Pi_A \in \T^{r, NF} (\AND, 1 / m^2)$ in Lemma~\ref{lem:QIC-DISJ-AND} and 
combining with Lemma~\ref{lem:QICmu0},
we get, for any $r \geq 1$,
\begin{align*}
	\QIC(\DISJ_m, 2/m)
	&\leq m  \cdot \Big(  \QIC_0^{r, NF} (\AND, 1/m^2) + H(w)\Big) + o (\sqrt{m}),
\end{align*}
where the l.h.s. is independent of $r$.
 Moreover, by Lemma~\ref{lem:LB-QIC-DISJ}, the left-hand side belongs to $\Omega (\sqrt{m})$, so by further combining with Lemma~\ref{lem:boosterror}, we can rewrite this as
\begin{align}
\Omega \left(\frac{1}{\sqrt{m} \lg m}\right) \leq \QIC^{r, NF} (\AND, 1/3).
\end{align}
The r.h.s. is independent of $m$, so by fixing $m$ to a large enough constant, we get, for any number of round $r$,
$$
	\QIC^{r, NF} (\AND, 1/3) \in \Omega (1).
$$
Hence, by Lemma~\ref{lem:QCC-DISJ-QIC-AND}, for any $n$,
$$
	\QCC^{r, NF} (\DISJ_n) \in \Omega (n),
$$
which concludes the proof of Theorem~\ref{thm:disj-NF-n}.



\section{Quantum Simulation of Classical Protocols}
\label{sec:quantum-simu-classical}

We now study how to quantumly simulate  classical protocols, and how the corresponding QIC behaves. By simulating, we mean that there is a quantum protocol with the same input-output behavior. 
It turns out that we can always find a quantum simulation with the same information cost as the classical protocol; it is even possible to build this quantum simulation such that it does not forget information.

For the reader's convenience, we deal successively with deterministic protocols, public coin protocols, and protocols with private coins. The latter needs a special care and we give a more detailed explanation on the construction.

\paragraph{Deterministic protocols.}
Let us consider a classical deterministic (i.e., which does not depend on private or shared randomness) protocol $\Pi$. We define the protocol $\Pi_0$ which is similar to $\Pi$ except that Alice and Bob keep local 
copies of their inputs
and of the messages, possibly padding messages with $0$'s such that the order of speech is known in 
advance to both
and independent of the inputs.
\begin{remark}
This might affect the communication cost of the protocol, but does not 
change the information
cost or the input-output behavior.
\end{remark}

Now, we define $\Pi_0^*$, the quantum simulation of $\Pi_0$ (hence it simulates $\Pi$ as well).
To generate their quantum messages, Alice and Bob run as unitary a classical 
reversible circuit implementing
the protocol in each round, and measure the output registers at the end.
\begin{lemma}
	The quantum simulation $\Pi_0^*$ has the same input-output behavior and information cost as the 
original deterministic protocol $\Pi$, and the same communication cost as the padded protocol $\Pi_0$.
\end{lemma}
The fact that the information cost is unchanged 
follows by noticing that each register is classical in HIC, which is equal to the IC of the 
classical protocol, and also HIC = CIC which are then also equal to QIC.

\paragraph{Public Coin Protocols.}
Let us now consider a classical protocol $\Pi$ with shared randomness. As above, we define  a classical protocol $\Pi_0$ similar to $\Pi$ where the players first make a local copy of the shared randomness, and then  pad their messages with $0$'s such that the
order of speech is known in advance to both, independently not only of the input, but also of the randomness.

Then we define the quantum simulation protocol $\Pi_0^*$ by having Alice and Bob use pure shared entanglement 
to simulate in a canonical way the shared randomness: make two coherent (quantum), perfectly correlated copies of the random strings,
a copy being given to Alice and the other one given to Bob. In this way, if either copy is traced out, the other copy is classical and
distributed exactly as the corresponding local copy of the shared randomness. 

Viewing a classical protocol with shared 
randomness as one which is an average over deterministic protocols with fixed random strings, they
can then run the corresponding classical deterministic protocol.
\begin{lemma}
The input-output behavior and the information cost of the quantum
simulation protocol $\Pi_0^*$
is the same as for the original public coin classical protocol $\Pi$, and the communication cost is the same as that of the padded protocol $\Pi_0$.
\end{lemma}
Once again, the fact that the information cost is unchanged 
follows by noticing that each register is classical in HIC, which is equal to the IC of the 
classical protocol, and also HIC = CIC which are then also equal to QIC.

\paragraph{Protocols with Private Randomness.}
The case of classical protocol that also have private randomness is the most tricky to handle. As a first attempt, the private randomness 
can be simulated
in a way similar to public randomness as described above, except that now both coherent copies of the random strings are 
given to the same party (the one who owns this private random string in the classical protocol). However, these registers do not 
look like classical  registers in the different information costs, and the above argument for classical protocols with only public randomness cannot be used to argue that the information remains unchanged.

Instead, we use a two-step procedure to obtain a quantum simulation protocol for which we can more easily show that the information cost is maintained. The first step consists in giving a classical simulation protocol of the original protocol in which the private randomness is in some canonical form. In the second step, we simulate quantumly this intermediate classical protocol by applying similar arguments as for classical protocols with only public randomness.

\underline{Step 1~: canonical classical simulation.} 
Consider a classical protocol $\Pi$. Let us first define a canonical transformation which provides another classical protocol, denoted $\widetilde \Pi$, in a particular form. 
For this canonical classical simulation, the idea is to  use a lot of fresh private 
randomness in each round, which directly encodes the distribution over messages in each round in a way 
which is consistent with the local information (input, shared randomness, and previous messages) of the sender. 
More precisely, say in round $i$ in $\Pi$, Alice is to generate message $M_i$ as a 
deterministic function of her input $X$, the shared randomness $R_{AB}$, her private randomness $S_A$, 
and the previous messages $M_{<i} = M_1 \dots M_{i-1}$.

For a given (partial) view $(x,r,m_{<i})$ of Alice at round $i$ (excluding her private randomness), consider the random variable $M_i^{x, r, m_{<i}}$ 
obtained by "averaging" the private randomness $s_A$, that is~: for any fixed message $m$,
$$
	\Pr[M_i^{x, r, m_{<i}} = m] = \PP_{S_A}[m_i(x, S_A, r, m_{<i}) = m].
$$
 Then the canonical simulation protocol $\widetilde \Pi$ 
uses in round $i$ the following random variable (which is given to Alice as fresh private randomness)~:
$$
	T_i^A = \bigotimes_{x, r, m_{<i}} M_i^{A,x, r, m_{<i}} \,,
$$ 
that is, independent copies of the random variable $M_i$ corresponding to each possible local view $(x, r, m_{<i})$. 
At round $i$, Alice considers her actual local view ($x, r, m_{<i}$), and 
sends the message corresponding to $M_i^{A,x, r, m_{<i}}$, that is, the element of her  
private randomness $T_i^A$ corresponding to her actual local view (the other parts of the private 
randomness $T_i^A$ are left untouched). Bob acts similarly, with some fresh private randomness $T^B_i$ at each even round $i$.
We denote $T^A = \otimes_{i \, odd} T_i^A$ and $T^B = \otimes_{i \, even} T_i^B$.

\begin{lemma}

	In this canonical classical simulation, both the information cost and the communication cost are unchanged~: for any input distribution $\mu$,
	$$
		IC(\widetilde \Pi, \mu) = IC(\Pi, \mu),
		\qquad CC(\widetilde \Pi) = CC(\Pi).
	$$
	Moreover,
the distribution of the joint random variable $XYRM_{\leq n}$ for the whole $n$-round protocol is also unchanged, 
and thus so is the input-output behavior.
\end{lemma}

\underline{Step 2~: quantum simulation.}
We consider a protocol $\widetilde \Pi_0$ in which the messages of $\widetilde \Pi$ are padded so that the order of speech is independent of the inputs and both public and private randomness. For the quantum simulation protocol $\widetilde \Pi_0^*$, private randomness is simulated by giving two coherent
 local copies to the player and letting him or her work on one of them.
 \begin{lemma}
 \label{lem:simu-quantum}
The input-output behavior and the information cost of the quantum
simulation protocol $\widetilde{\Pi}_0^*$
is the same as for the original classical protocol $\Pi$ with private randomness, and the communication cost is the same as that of the padded protocol $\widetilde{\Pi}_0$.

\end{lemma}
\begin{proof}
We first focus on the CIC term. Consider for instance the third round (Alice is the sender). Dropping the ancilla qubits for brevity, the global quantum state just after Bob receives the third message is then~:
$$
	\rho_3 = \rho^{X, R_X, R_{AB}^A, T^A, M^A_{\leq 3}, M^M_3, Y, R_Y, R_{AB}^B, T^B, M^B_{<3}}
$$
where $M^A_i$ and $M^B_i$ denote respectively Alice and Bob's copy of the $i$-th message, whereas $M^M_3$ is the register that is sent over from Alice to Bob. The third term appearing in CIC is~:
\begin{align}
\label{quantsimu:eq:thirdTermCIC}
	I(M^M_3 : X | Y, R_{AB}^B, T^B, M^B_{<3}),
\end{align}
where the CQMI is evaluated on the quantum state~:
\begin{align*}
	&\rho^{X, M^M_3, Y, R_{AB}^B, T^B, M^B_{<3}}
	\\
	= \,& \Tra{R_X, R_{AB}^A, T^A, M^A_{\leq 3}, R_Y} (\rho_3)
	\\
	= \,& \sum_{x,y,r, m_{\leq 3}} p(x,y,r, m_{\leq 3}) \ket{x,y,r, m_{\leq 3}}\bra{x,y,r, m_{\leq 3}} \otimes \rho^{x, m_3, y,r, T^B, m_{<3}}
\end{align*}
for some family of quantum states $(\rho^{x, m_3, y,r, T^B, m_{< 3}})_{y,r, m_{\leq 3}}$. For the last equality, we used the fact that the registers $X, M^M_3, Y, R_{AB}^B, M^B_{<3}$ are in a classical state, since the registers  $R_X , R_{AB}^A, T^A, M^A_{\leq 3}, R_Y$ are traced out. Furthermore, recall that in the classical protocol $\widetilde \Pi$, the random variable $T^B$ is defined as~:
\begin{align*}
	T^B 
	&= T^B_2 \otimes \left( \bigotimes_{i \geq 2} T^B_{2i} \right)
	= \left( \bigotimes_{y, r, m_{1}} M_2^{B,y, r, m_{1}} \right) \otimes \left( \bigotimes_{i \geq 2} T^B_{2i} \right).
\end{align*}
In the third round of the quantum protocol, since the registers $R_{AB}^A, M^A_{\leq 3}, Y'$ are already traced out, the quantum state can actually be decomposed as~:
\begin{align*}
	&\rho^{X, M^M_3, Y, R_{AB}^B, T^B, M^B_{<3}}
	\\
	= \,& \left(\sum_{x,y,r, m_{\leq 3}} p(x,y,r, m_{\leq 3}) \ket{x,y,r, m_{\leq 3}}\bra{x,y,r, m_{\leq 3}} \otimes \rho^{x, m_3, y,r, T_2^B, m_{<3}}\right)
	\otimes \left(\bigotimes_{i \geq 2} \rho^{ T^B_{2i}} \right).
\end{align*}
Hence, by Lemma~\ref{prelim:lem:CQMIfacts}, the term~\eqref{quantsimu:eq:thirdTermCIC} can be written
\begin{align}
	I(M^M_3 : X | Y, R_{AB}^B, T^B, M^B_{<3}) = \bbE_{y, r, m_{<3}}\left[ I(M^M_3 : X | T_2^B )_{\rho^{X, M^M_3, y, r, T_2^B, m_{<3}}}\right],
\end{align}
with 
$$
	\rho^{X, M^M_3, y, r, T_2^B, m_{<3}} = \sum_{x, m_{3}} p(x, m_{3}) \ket{x, m_{3}}\bra{x, m_{3}} \otimes \rho^{x, m_3, y,r, T_2^B, m_{<3}},
$$
where we use the fact that $X$ and $M^M_3$ are classical since $R_X$ and $M^A_3$ were traced out. The $T_2^B$ is still quantum, but it has a special structure: either $M_2^{B, y^\prime r^\prime m_1^\prime}$ does not correspond to the actual view $(y, r, m_1)$ of Bob, and so it remains in a pure state, or else it corresponds but Alice possesses a coherent copy of $M_2^{B, y, r, m_1}$, and so Bob's copy is classical once we trace Alice's copy out.
It follows that 
\begin{align}
	I(M^M_3 : X | Y, R_{AB}^B, T^B, M^B_{<3}) & = \bbE_{y, r, m_{<3}}\left[ I(M^M_3 : X | M_2^{B, y, r, m_1} )_{\rho^{X, M^M_3, y, r, M_2^{B, y, r, m_1}, m_{<3}}}\right] \\
	& = I(M^M_3 : X | M_2^{B, Y, R, M_1} ),
\end{align}
as in $IC (\widetilde \Pi_0)$.

More generally, consider an odd round $i$ (Bob is the receiver). We can see that, conditioning on the classical part $(y, r, m_{<i})$, 
all of the quantum registers corresponding to the private randomness $T^B$ on Bob's side fall into two categories~:
\begin{itemize}
	\item either they have never been used (for $j \geq i$, all of $T^B_j$, or for $j \leq i$, the coordinates of $T^B_j$ which did not correspond to the actual view of Bob at round $j$), and so remain  in
a pure state in product form and can be eliminated from the CQMI term,
	\item or else they have been used but correspond to one of at least some quantum copies of a message previously sent to the other party (the coordinates of $T^B_j$ for $j \leq i$, $j$ odd, corresponding to the local view $(y, r, m_{<j})$ of Bob at round $j$, hence to a message $M_j^{B,x, r, m_{<j}}$ sent by Bob to Alice). In the CQMI, since one party's registers are traced out, this term of CIC is classical.
\end{itemize}
Using the chain rule, we see that the $i$-th term in CIC for the quantum simulation is equal to the $i$-th term in the information cost of the classical protocol $\widetilde \Pi$.
Similar arguments hold also for any even round. Hence $\CIC(\widetilde \Pi_0^*, \mu) = IC(\widetilde \Pi, \mu)$.
Finally, we can see that $\HIC(\widetilde \Pi_0^*, \mu) = \CIC(\widetilde \Pi_0^*, \mu)$ by using the chain rule in an order so as to be able to apply the above argument to the quantum registers corresponding to private randomness.
This implies $\CRIC (\widetilde \Pi_0^*, \mu) = 0$, 
and $\QIC(\widetilde \Pi_0^*, \mu) = \CIC(\widetilde \Pi_0^*, \mu) = \HIC(\widetilde \Pi_0^*, \mu) = IC(\widetilde \Pi_0, \mu) = IC(\Pi, \mu)$.
\end{proof}

\begin{remark}

In particular for classical protocols, $IC_0 (\AND) \in \Omega (1)$ (and then also $CC (\DISJ_n) \in \Omega (n)$ by a standard direct sum argument akin to Lemma~\ref{lem:QCC-DISJ-QIC-AND}) follow by using such a quantum 
simulation that does not forget information and using the result $QIC^{r, NF} \in \Omega (1)$ from the previous section.
Surprisingly, the main ingredients going into this proof of the linear lower bound on the classical communication complexity of disjointness are a $\sqrt{n}$  lower bound on the quantum information complexity and a $\sqrt{n}$ upper bound on the quantum communication complexity of disjointness, two $\sqrt{n}$ bounds.

\end{remark}


\section{Clean Protocols, IP, and Random Functions}
\label{sec:IPrandfct}

\subsection{Clean Protocols and Phase Encoding of the Output}

The development in this section follows that of Refs~\cite{CvDNT99, MW07}. The Information Flow Lemma (see Lemma~\ref{lem:rdindeplb}) allows us to translate their arguments about QCC to QIC. The link with IC follows by the general simulation procedure of classical protocols maintaining IC (see Lemma~\ref{lem:simu-quantum}).

Given a Boolean function $f$ and any protocol $\Pi$ computing $f$ with zero-error, 
we will construct a so-called \emph{clean} protocol $\Pi^\prime$ also computing $f$ 
with zero-error, but restoring all registers, except for an output qubit, to their original state.
Then, using similar ideas, we define a protocol $\Pi^{\prime \prime}$ where the output is in the phase.

\paragraph{Clean protocol $\Pi^\prime$.}

The action of $\Pi$, if we do not trace out the $A^\prime$, $B^\prime$ registers, is given by the sequence of unitaries $U_{1}, U_{2}, \dots, U_{r}, U_{r+1}$ applied by Alice and Bob in turns. Hence, denoting $U_\Pi =  U_{r+1} U_r \cdots U_2 U_1$, the state at the end of a run of $\Pi$ on input $(x, y)$ is of the form
\begin{align}
	& U_\Pi \Big(\ket{x}^X \ket{y}^Y \ket{\psi}^{T_A T_B}\Big)
	\nonumber
	\\
	& = \, \ket{x}^X \ket{y}^Y \ket{f (x, y)}^{B_{out}} \ket{\phi_{xy}}^{A_{out} A^\prime B^\prime},
	\label{randfct:eq:afterPi}
\end{align}
for some state $\ket{\phi_{xy}}$ depending on both $x$ and $y$.

We define the protocol $\Pi^{\prime}$ as the protocol whose global action is given by $U_\Pi^\dagger CNOT_{B_{out} \rightarrow B_{out}^\prime} U_\Pi$ which uses an additional ancillary qubit $\ket{0}^{B_{out}^\prime}$.
In other words, the players start by running $\Pi$, which leads to the state~\eqref{randfct:eq:afterPi}.
Then, Bob applyies a $CNOT$ 
gate from $B_{out}$ to $B_{out}^\prime$, which gives the state

\begin{align}
CNOT_{B_{out} \rightarrow B_{out}^\prime} \Big(\ket{x}^X & \ket{y}^Y  \ket{f (x, y)}^{B_{out}} \ket{\phi_{xy}}^{A_{out} A^\prime B^\prime} \ket{0}^{B_{out}^\prime}\Big)
	\nonumber
	 \\
	& = \ket{x}^X \ket{y}^Y \ket{f (x, y)}^{B_{out}} \ket{\phi_{xy}}^{A_{out} A^\prime B^\prime} \ket{f(x, y)}^{B_{out}^\prime}.
	\nonumber
\end{align}

To clean the working registers, the players run the protocol whose action is $U_\Pi^\dagger$, and they obtain

\begin{align}
U_\Pi^\dagger \Big(\ket{x}^X & \ket{y}^Y \ket{f (x, y)}^{B_{out}} \ket{\phi_{xy}}^{A_{out} A^\prime B^\prime} \Big) \ket{f(x, y)}^{B_{out}^\prime}
	\nonumber
	\\
	& = \ket{x}^X \ket{y}^Y \ket{\psi}^{T_A T_B} \ket{f(x, y)}^{B_{out}^\prime}.
	\nonumber
\end{align}

So the overall action of $\Pi^\prime$ is
\begin{align}
U_\Pi^\dagger CNOT_{B_{out} \rightarrow B_{out}^\prime} & U_\Pi \Big(\ket{x}^X  \ket{y}^{Y} \ket{\psi}^{T_A T_B} \ket{0}^{B_{out}^\prime}\Big) \\
& = \ket{x}^X \ket{y}^Y \ket{\psi}^{T_A T_B} \ket{f(x, y)}^{B_{out}^\prime}.
\end{align}

\begin{remark}
Notice that if $QCC_{A \rightarrow B} (\Pi) = a$ (the communication from Alice to Bob), 
$QCC_{B \rightarrow A} (\Pi) = b$ (the communication from Bob to Alice), 
then $QCC_{A \rightarrow B} (U_\Pi^\dagger) = b $, $QCC_{B \rightarrow A} (U_\Pi^\dagger) = a$; 
hence $QCC_{A \rightarrow B} (\Pi^\prime) = a+ b = QCC(\Pi)$.
We will later argue something similar for information of zero-error protocols.
\end{remark}

\paragraph{Protocol $\Pi^{\prime \prime}$ with output in the phase.}

We define $\Pi^{\prime \prime}$ similarly to $\Pi^\prime$, except that the ancilla register $B_{out}^\prime$ is originally in the state
$
	\ket{-}  = \frac{1}{\sqrt{2}} (\ket{0} - \ket{1}),
$
instead of $\ket{0}$ as in $\Pi^\prime$.
As a consequence, instead of recording $f(x, y)$ in the computational basis of $B_{out}^\prime$, the players ``record'' it in the phase. 
So, after running $\Pi$, the players apply $CNOT_{B_{out} \rightarrow B_{out}^\prime}$ to obtain

\begin{align}
&CNOT_{B_{out} \rightarrow B_{out}^\prime} \left(\ket{x}^X \ket{y}^Y \ket{f (x, y)}^{B_{out}} \ket{\phi_{xy}}^{A_{out} A^\prime B^\prime} \ket{-}^{B_{out}^\prime}\right) 
	\nonumber
	\\
	= \,& \ket{x}^X \ket{y}^Y \ket{f (x, y)}^{B_{out}} \ket{\phi_{xy}}^{A_{out} A^\prime B^\prime} (-1)^{f(x, y)} \ket{-}^{B_{out}^\prime}.
	\nonumber
\end{align}

Thus, running $U_\Pi^\dagger$ and bringing out the phase, we get

\begin{align}
&(-1)^{f(x, y)} U_\Pi^\dagger \left( \ket{x}^X \ket{y}^Y \ket{f (x, y)}^{B_{out}} \ket{\phi_{xy}}^{A_{out} A^\prime B^\prime} \right) \ket{-}^{B_{out}^\prime} 
	\nonumber
	\\
	= \, &(-1)^{f(x, y)} \ket{x}^X \ket{y}^Y \ket{\psi}^{T_A T_B} \ket{-}^{B_{out}^\prime}.
	\nonumber
\end{align}

\subsection{Relating to $QIC (\Pi, \mu)$}

These two protocols, $\Pi^{\prime}$ and $\Pi^{\prime \prime}$, have the same communication cost, and in particular: 
$$QCC_{A \rightarrow B} (\Pi^{\prime \prime}) = QCC_{A \rightarrow B} (\Pi^{\prime}) = QCC (\Pi).$$
We now study their information cost and show the following result.

\begin{proposition}
For any input distribution $\mu$, Boolean function $f$, and
 any zero-error protocol $\Pi$ for $f$, 
$$QIC (\Pi, \mu) = QIC_{A \rightarrow B} (\Pi^\prime, \mu) = QIC _{A \rightarrow B} (\Pi^{\prime \prime }, \mu).$$
\end{proposition}

\begin{proof}
It is clear for the first half of $\Pi^{\prime}$ and $\Pi^{\prime \prime}$, when running $U_\Pi$ forward, that the corresponding information costs are $QIC_{A \rightarrow B} (\Pi, \mu)$. We now argue that for the second half, when running $U_\Pi^\dagger$, the corresponding information cost is $QIC_{B \rightarrow A} (\Pi, \mu)$.

Consider first the clean protocol $\Pi^\prime$, and view register $B_{out}^\prime$,
containing a copy of $f(x, y)$ while running $U_{\Pi}^\dagger$, as an additional part of a purification register
$R^\prime$ for protocol $\Pi$: $R^{\prime} = R_X^\prime R_Y^\prime B_{out}^\prime$. This is justified as follows. We can instead think of $B_{out}^\prime$ as being generated, after running $U_\Pi$ with $\ket{0}$ in $B_{out}^\prime$ (and thus registers $R_X^\prime R_Y^\prime = R_X R_Y$ purify registers $XY A_i B_i C_i$ for that part), by applying $U_f$, defined such that $U_f (\ket{x} \ket{y} \ket{0} ) = \ket{x} \ket{y} \ket{f(x, y)}$, to the registers $R_X^\prime R_Y^\prime B_{out}^\prime$. Since $\Pi$, and hence $\Pi^\prime$, is a zero error protocol, and $f$ is a function, the resulting states at that point and when further applying $U_\Pi^\dagger$ will then be the same in this modified purified view of $\Pi$ as in the clean protocol $\Pi^\prime$, and thus the $QIC$'s are also the same. But then the $QIC$'s are also identical  to the ones in which we run $U_\Pi$ and then $U_\Pi^\dagger$ without making a copy in $B_{out}^\prime$, since the global states are the same up to  unitary $U_f$ being appied to purification registers $R_X R_Y$ and the uncorrelated state $\ket{0}$ in $B_{out}^\prime$.

We can apply a similar reasoning to the protocol $\Pi^{\prime \prime}$ in which $f(x,y)$ is recorded in the phase, since we can similarly think of $(-1)^{f(x, y)} \ket{-}^{B_{out}^\prime}$ as being part of the reference $R^{\prime \prime} = R_X^{\prime \prime} R_Y^{\prime \prime} B_{out}^\prime$, with $B_{out}^\prime$ remaining in state $\ket{-}$, and the phase information now being generated by applying $U_{f, phase}$, defined such that $U_{f, phase} (\ket{x} \ket{y}) = (-1)^{f(x, y)} \ket{x} \ket{y}$, on registers $R_X R_Y$.

Finally, notice that if we run $U_\Pi$ and then $U_\Pi^\dagger$ without acting on the output, then, using the duality relation $I (R_X R_Y ; C_i |  Y B_i) = I (R_X R_Y ; C_i | X A_i) $, we get that $QIC_{A \rightarrow B}$ of $\Pi^\prime$  and $\Pi^{\prime \prime}$ in the $U_\Pi^\dagger$ part is $QIC_{B \rightarrow A} (\Pi, \mu)$, so
\begin{align*}
QIC_{A \rightarrow B} (\Pi^{\prime \prime }, \mu) & = QIC_{A \rightarrow B} (\Pi^{\prime }, \mu) \\
	& = QIC_{B \rightarrow A} (\Pi^{\prime \prime }, \mu) \\
	& = QIC_{B \rightarrow A} (\Pi^{\prime  }, \mu) \\
	& = QIC_{A \rightarrow B} (\Pi, \mu) + QIC_{B \rightarrow A} (\Pi, \mu) \\
	& = QIC (\Pi, \mu).
\end{align*}
\end{proof}

\subsection{Information Lower Bound}

To get a tractable lower bound 
on $QIC_{A \rightarrow B} (\Pi^{\prime \prime}, \mu) = QIC(\Pi, \mu)$, we focus
on total functions and on product distributions $\mu = \mu_X \otimes \mu_Y$ on $XY$, and
we apply the Information Flow Lemma. Taking the purified view, we have in $\Pi^{\prime \prime}$
\begin{align}
	\ket{\rho_{f, \mu}^\prime}^{X R_X Y R_Y T_A T_B B_{out}^\prime}
	& =
	\left(\sum_{x, y} (-1)^{f(x, y)}\sqrt{\mu_X (x)} \sqrt{\mu_Y (y)} \ket{xxyy}^{XR_X Y R_Y} \right)\ket{\psi}^{T_A T_B} \ket{-}^{B_{out}^\prime},
\end{align}
in which we emphasize the dependance of $\ket{\rho_{f, \mu}^\prime}$ on the function $f$ and the product distribution $\mu$.
\begin{proposition}
\label{prop:QICgeqCMI}
We have the following lower bound:
	$$
		QIC (\Pi , \mu)
		\geq
		I (R_X ; Y  R_Y)_{\rho_{f, \mu}^\prime}.
	$$
\end{proposition}
\begin{proof}
Notice that $B_{out}^\prime$ remains in the pure state $\ket{-}$ throughout, independently of $x$ and $y$, and we can remove that register from all the information terms below.
We have successively the following chain:
\begin{align}
			QIC (\Pi , \mu) = \,& QIC_{A \rightarrow B} (\Pi^{\prime \prime} , \mu)
			\nonumber
			\\
			\geq \,& \sum_{i~odd} I(R_X R_Y ; C_i | Y B_i)_{\rho_{i, \mu}^\prime}
			\nonumber
			\\
			\geq \,& \sum_{i~odd} I(R_X ; C_i | R_Y Y B_i)_{\rho_{i, \mu}^\prime}
			\nonumber
			\\
			\geq \,& \sum_{i~odd} I(R_X ; C_i | R_Y Y B_i)_{\rho_{i, \mu}^\prime} - \sum_{i~even} I (R_X ; C_i | R_Y Y B_i)_{\rho_{i, \mu}^\prime}
			\nonumber
			\\
			= \,& I (R_X ; Y T_B B_{out}^\prime | R_Y)_{\rho_{f, \mu}^\prime} - I(R_X ; Y | R_Y)_{\rho_{0, \mu}^\prime} 
			\label{randfct:eq:appliCoroPiprime}
			\\
			= \,& I (R_X ; Y T_B B_{out}^\prime | R_Y)_{\rho_{f, \mu}^\prime}
			\label{randfct:eq:ICproduct0}
			\\
			= \,& I (R_X ; Y  R_Y)_{\rho_{f, \mu}^\prime} \,,
			\label{randfct:eq:ICproduct02}
\end{align}
where equality~\eqref{randfct:eq:appliCoroPiprime} is obtained by application of the Information Flow Lemma under the form of Corollary~\ref{lem:rdindeplb-coro}, with $E_1 = R_X$, $E_2 = R_Y$ in $\Pi^{\prime \prime}$.
Equality~\eqref{randfct:eq:ICproduct0} holds since $\mu = \mu_X \otimes \mu_Y$ is a product distribution, 
so $\ket{\rho_{0, \mu}^\prime} = \left(\sum_x \sqrt{\mu_X (x)} \ket{xx}^{X R_X}\right) \left(\sum_y \sqrt{\mu_Y (y)} \ket{yy}^{Y R_Y}\right).$
As for equality~\eqref{randfct:eq:ICproduct02}, notice first that $I(R_X; Y T_B B_{out}^\prime  | R_Y)_{\rho_{f, \mu}^\prime} = I(R_X ; Y | R_Y)_{\rho_{f, \mu}^\prime}$ (with the same arguments as above). 
Moreover, $I(R_X ; R_Y)_{\rho_{f, \mu}^\prime} = I(X ; Y)_{\rho_{f, \mu}^\prime} = 0$, since this is a classical product state on $XY$, so $I(R_X ; Y | R_Y)_{\rho_{f, \mu}^\prime} = I(R_X ; Y  R_Y)_{\rho_{f, \mu}^\prime}$.
\end{proof}

Contrasting $\rho_{f, \mu}^\prime$ to $\rho_{0, \mu}^\prime$, if the $Y R_Y$ registers contain information about the $X$ register, it must be ``encoded in the phase'' $(-1)^{f(x, y)}$ somehow. 
Another way to think about it, in the spirit of what was done in~\cite{CvDNT99, MW07}, is as follow: Alice is given a classical random 
variable $X$ distributed according to $\mu_X$, Bob locally prepares 
the pure state $\sum_y \sqrt{\mu_Y} \ket{y}^Y \ket{y}^{R_Y}$ 
corresponding to $\mu_Y$, and Alice and Bob run $\Pi^\prime$. 
Bob ends up with registers $YR_Y$ (and $T_B B_{out}^\prime$, 
which were restored to state $\ket{\psi}^{T_A T_B} \ket{-}^{B_{out}^\prime} $, 
independent of $x$) of $\rho_{f, \mu}^\prime$, which we now view as the output 
of a ``noisy'' classical-quantum communication channel with input register $X$, 
in which the different phases allows to (at least partially, depending on $f$ and $\mu_Y$) distinguish the pure 
states 
\begin{align}
\ket{\rho_{f, x, \mu_Y}^\prime}^{Y R_Y} = \sum_y (-1)^{f(x, y)} \sqrt{\mu_Y} \ket{y}^Y \ket{y}^{R_Y},
\end{align} 
corresponding to each $x$. The (channel) Holevo information $\max_{\mu_X} I(X ; Y R_Y)_{\rho_{f, \mu}^\prime}$ is a known asymptotically 
achievable bound for classical communication over such noisy channels, giving an alternate proof 
sketch of $I (X ; Y R_Y)_{\rho_{f, \mu}^\prime} \leq 2 QCC_{A \rightarrow B} (\Pi^\prime)$ 
(also using the optimality of super-dense coding; the factor of two disappear if the messages 
are classical, and also if we do not allow for pre-shared entanglement in $\Pi$).

Now, 
$$I (X ; Y R_Y)_{\rho_{f, \mu}^\prime} = H (Y R_Y)_{\rho_{f, \mu}^\prime} - H (Y R_Y | X)_{\rho_{f, \mu}^\prime}\,,$$
and
$$H (Y R_Y | X)_{\rho_{f, \mu}^\prime} = \bbE_X H (Y R_Y)_{\rho_{f, x, \mu_Y}^\prime} = 0 \,,$$ 
since $\ket{\rho_{f, x, \mu_Y}^\prime}^{Y R_Y}$ is a pure state for each $x$.
Notice that $I (R_X ; Y R_Y)_{\rho_{f, \mu}^\prime} = H (Y R_Y)_{\rho_{f, \mu}^\prime}$ only depends on $\mu_X$, $\mu_Y$, and $f$:

\begin{align}
\rho_{{f, \mu}}^{\prime Y R_Y} = \sum_x \mu_X (x) \kb{\rho_{f, x, \mu_Y}^\prime}{\rho_{f, x, \mu_Y}^\prime}^{Y R_Y}.
\end{align}
From Proposition~\eqref{prop:QICgeqCMI} and the discussion above we obtain the following lower bound:
\begin{equation}
\label{eq:QICgeqH}
		QIC (\Pi , \mu)
		\geq
		H (Y R_Y)_{\rho_{f, \mu}^\prime}.
\end{equation}

\subsection{Inner Product function}

The case of the Inner Product function was studied using a similar argument in Ref.~\cite{CvDNT99}.
Let us consider $f (x, y) = IP_n (x, y) = x \cdot y$ on $\lg|X| = \lg |Y| = n$ bits, and take $\mu_X = \mu_Y$ 
the uniformly random distribution. If Bob is given register $R_Y$ together with $Y$ of $\rho_{f, x, \mu_Y}^\prime$ and 
applies first $(CNOT^{\otimes n})_{Y \rightarrow R_Y}$ and then $H^{\otimes n}$ on $Y$, he 
gets, for any fixed $x$ on Alice's side,

\begin{align}
&(H^{\otimes n})^Y (CNOT^{\otimes n})_{Y \rightarrow R_Y} \left( 2^{-n/2} \sum_y (-1)^{x \cdot y} \ket{yy}^{Y R_Y}\right)
		\nonumber
		\\ 
		= \, & (H^{\otimes n})^Y  \left(2^{-n/2} \sum_y (-1)^{x \cdot y} \ket{y}^{Y }\right) \ket{0^n}^{R_Y} 
		\nonumber
		\\
		= \, & \ket{x}^Y \ket{0^n}^{R_Y},
		\nonumber
\end{align}
since $H^{\otimes n}$ is self-inverse and $H^{\otimes n} \ket{x} = 2^{-n/2} \sum_y (-1)^{x \cdot y} \ket{y}$.
By isometric invariance of von Neumann entropy, $H(Y R_Y)_{\rho_{f, \mu}^\prime} = H ( X^\prime) = n$, for $X^\prime$ 
a classical copy of $X$. We get 
$$\QIC (IP_n, \nu,  0) \geq n,$$
with $\nu$ the uniform distribution on the inputs.
Since we only assumed that Bob can compute the function value in our lower bound, we get a matching upper bound for such protocols, and $\QIC (IP_n, \nu, 0)= n$.

\subsection{Random Functions}

The argument of Ref.~\cite{CvDNT99} for the IP function was extended in Ref.~\cite{MW07} to the study of arbitrary (total) Boolean function, 
and in particular to argue about the quantum communication complexity of a random Boolean function. 
They showed, for $\nu$ the uniform distribution on $n + n$ bit inputs (i.e.~$\lg|X| = \lg |Y| = n$), that a uniformly random Boolean function $f$, 
(a function chosen by picking $f(x, y)$ uniformly at random in $\{ 0, 1\}$ 
for each pair $(x, y)$), satisfies with high probability $H (Y R_Y)_{\rho_{f, \mu}^\prime} \geq n (1 - o(1))$, 
and thus $\QCC (f, \nu, 0) \geq n (1 - o(n))$. Moreover, for small enough constant $\epsilon > 0$, 
they  also show using a continuity argument that $\QCC (f, \nu, \epsilon) \in \Omega (n)$. Thus, most Boolean functions 
have essentially a linear quantum communication complexity.

We focus on the case $\epsilon = 0$, and extend their results for $QIC$ of a random function.
We use the following result proved in Ref.~\cite{MW07}. Here, $H_2$ is the R\'enyi entropy of order $2$, $\nu$ is the uniform distribution on $2n$-bit strings, and the probability is taken over the random choice of $f$, also picked uniformly at random in $\{0, 1 \}$ for each of the $2^{2n}$ pairs $(x, y)$.

\begin{theorem}
\begin{align*}
\Pr_f \big[H_2 (Y R_Y)_{\rho_{f, \nu}^\prime} < (1 - \delta) n \big] \leq \exp (- (2^{\delta n} - 1)^2 / 2),
\end{align*}
where the probability is uniform over Boolean functions of $n+n$ bits.
\end{theorem}

Since $H_2 \leq H$, we get the following theorem by taking $\delta = 1 / \sqrt{n}$ above and using~\eqref{eq:QICgeqH}.
\begin{theorem}
\begin{align*}
\Pr_f \big[QIC (f, \nu, 0) < (1 - 1/ \sqrt{n})n \big] \leq  \exp (- (2^{\sqrt{n}} - 1)^2 / 2),
\end{align*}
where the probability is uniform over Boolean functions of $n+n$ bits.
\end{theorem}

Hence, except with negligible probability over the choice of a random function $f$,
$$QIC (f, \nu, 0) \geq n (1 - o(1)).$$

\subsection{Non-Zero Error and Classical Protocols}

Using the quantum simulation of classical protocols maintaining the classical IC that we gave in Section~\ref{sec:quantum-simu-classical}, the above result also implies a bound for any classical protocol. Moreover, it is known (see Ref.~\cite{MR3210776-infoExactComm}) that classical IC is continuous at $\epsilon = 0$, so we get the following corollary.
\begin{corollary}
\begin{align*}
\Pr_f \left[\lim_{\epsilon \rightarrow 0} IC (f, \nu, \epsilon) < (1 - 1/ \sqrt{n})n \right] \leq  \exp (- (2^{\sqrt{n}} - 1)^2 / 2),
\end{align*}
where the probability is uniform over Boolean functions of $n+n$ bits.
\end{corollary}

Hence, except with negligible probability over the choice of a random function $f$, we have 
$$\lim_{\epsilon \rightarrow 0} IC (f, \nu, \epsilon) \geq n (1 - o(1)). $$
To the best of our knowledge, this is the first proof that $\lim_{\epsilon \rightarrow 0} IC$  for a random function is essentially $n$, and not only $\Omega (n)$, which was known at least since the work of Braverman and Weinstein~\cite{MR3003574-discLBIC} proving that discrepancy lower bounds IC (through a compression argument).

It is an important open question to determine whether it also holds that QIC is continuous at $\epsilon = 0$, which would then imply a similar result in the quantum setting.


\paragraph{Acknowledgments. }
The authors are very grateful to Anurag Anshu, Andr{\'e} Chailloux, Ankit Garg, Iordanis Kerenidis, Ashwin Nayak, and Penghui Yao for many useful discussions.
M.L.~has been supported by ERC grant QCC. D.T.~is supported in part by NSERC, CIFAR, Industry Canada and ARL CDQI program. IQC and PI are supported in part by the Government of Canada and the Province of Ontario.
Part of this research was conducted while M.L. was a PhD student with the Institut de Recherche en Informatique Fondamentale, Universit{\'e} Paris Diderot, and while D.T.~was a PhD student with the D{\'e}partement d'informatique et de recherche op{\'e}rationnelle, Universit{\'e} de Montr{\'e}al and was supported in part by a FRQNT B2 Doctoral research scholarship, and by CryptoWorks21.


\bibliographystyle{alpha}
\bibliography{infoflow-0-bib}

\appendix


\section{Proofs for Section~\ref{sec:disj}}

\subsection{Proof of Lemma~\ref{lem:QCC-DISJ-QIC-AND}}\label{subsec:proof-lem-QCC-DISJ-QIC-AND}
Let us start by stating two intermediate lemmas that can be proved respectively as Lemma~4.18 and Lemma~4.19 in~\cite{BGKMT15}.
\begin{lemma}\label{lem:QICandQICdisj}
	For any integers $n, r,$ any $\epsilon>0$, and any input distribution $\mu_0$ such that $\mu_0(1,1) = 0$,
	$$
		\inf_{\Pi_A \in \T^{r,NF} (\AND, \epsilon)} \QIC^{r, NF}(\Pi_A, \mu_0)
		\leq
		\inf_{\Pi_D \in \T^{r,NF} (\DISJ_n, \epsilon)} \frac{1}{n} \QIC^{r, NF}(\Pi_D, \mu_0^{\otimes n}).
	$$
\end{lemma}
\begin{lemma}\label{lem:QICand-maxinf}
	$$
		\QIC^{r, NF}_0 (\AND, \epsilon) = \max_{\mu_0\,:\, \mu_0(1,1) = 0} \, \inf_{\Pi_A \in \T^{r,NF} (\AND_n, \epsilon)} \QIC(\Pi, \mu_0).
	$$
\end{lemma}

We can now proceed to the proof of Lemma~\ref{lem:QCC-DISJ-QIC-AND}.

\begin{proof}[Proof of Lemma~\ref{lem:QCC-DISJ-QIC-AND}]
	The result is a consequence of the following chain of inequalities:
	\begin{align*}
	\QCC^{r, NF}(\DISJ_n, 1/3)
	&\geq \QIC^{r, NF}(\DISJ_n, 1/3)
	\\
	&\geq   \max_\mu \inf_{\Pi_D \in  \T^{r, NF} (\DISJ_n, 1/3)} \QIC(\Pi_D, \mu)
	\\
	&\geq \max_{\mu_0\,:\, \mu_0(1,1) = 0} \, \inf_{\Pi_D \in \T^{r,NF} (\DISJ_n, 1/3)} \QIC(\Pi_D, \mu_0^{\otimes n}) 
	\\
	&\geq \max_{\mu_0\,:\, \mu_0(1,1) = 0} \, \inf_{\Pi_A \in \T^{r,NF} (\AND, 1/3)} n \cdot \QIC(\Pi_A, \mu_0) 
	\\
	&\geq  n \cdot \QIC^{r, NF}(\AND, 1/3).
	\end{align*} 
	The first inequality holds since $\QIC$ lower bounds $\QCC$, the second since the protocol can now be optimized according to $\mu$, the third since, on the r.h.s. the maximization is over a smaller set of product distributions satisfying $\mu_0(1,1) = 0$. The fourth is by Lemma~\ref{lem:QICandQICdisj}, and the last is by Lemma~\ref{lem:QICand-maxinf}.
\end{proof}

\subsection{Proof of Lemma~\ref{lem:QICmu0}}\label{subsec:proof-lem-QICmu0}

As a first step, we show that the second inequality of Lemma~4.7 in~Ref.~\cite{BGKMT15} admits a tighter version for protocols not forgetting information (according to Definition~\ref{def:protnoforget}).
	\begin{lemma}[Quasi-convexity in input]\label{lem:quasi-conv-improved}
		Let $p \in [0,1]$, and  $\mu_1, \mu_2$ be two input distribution. Define $\mu = p\mu_1 + (1-p) \mu_2$. Then the following holds for any protocol $\Pi$ which does not forget information:
		\begin{align*}
			&\QIC(\Pi,\mu) \geq p\QIC(\Pi, \mu_1) + (1-p) \QIC(\Pi, \mu_2),
			\\
			&\QIC(\Pi,\mu) \leq  p\QIC(\Pi, \mu_1) + (1-p) \QIC(\Pi, \mu_2) + H(p),
		\end{align*}
		independently of the number of rounds in $\Pi$.
	\end{lemma}
	Compared with  Lemma~4.7~Ref.~\cite{BGKMT15}, we save a multiplicative factor equals to the number of rounds in front of the term $H(p)$. 
	
	\begin{proof}[Proof of Lemma~\ref{lem:quasi-conv-improved}]
		The first inequality holds by the first inequality of Lemma~4.7 in~Ref.~\cite{BGKMT15}.
		Let us prove here the second inequality. Since $\Pi$ does not forget information, by Remark~\ref{charac:rem:equivNF}, its QIC is equal to its HIC. So it is sufficient to prove the desired inequality with QIC replaced by HIC. Let $R$ be a register holding a purification of $\rho_{\mu_1}$ and $\rho_{\mu_2}$. Then, we can purify $\rho_\mu$ with two copies $S_1,S_2$ of a selector reference register, such that
		$$
			| \rho_\mu \rangle ^{A_{in} B_{in} R S_1 S_2} = \sqrt{p} \, | \rho_{\mu_1} \rangle ^{A_{in} B_{in} R} |1\rangle^{S_1} |1\rangle^{S_2} + \sqrt{1-p} \, | \rho_{\mu_2} \rangle ^{A_{in} B_{in} R} |2\rangle^{S_1} |2\rangle^{S_2}.
		$$
		We can expand the HIC from Alice to Bob as:
		\begin{align*}
		\HIC_{A \rightarrow B}(\Pi, \mu)
			&= I(X; B_{out}B' | Y)_{\rho_\mu}
			\\
			&= I(X S_1; B_{out}B' | Y)_{\rho_\mu}
			\\
			&= I(S_1; B_{out}B' | Y)_{\rho_\mu} + I(X; B_{out}B' | Y S_1)_{\rho_\mu}
			\\
			&\leq H(p) + I(X; B_{out}B' | Y S_1)_{\rho_\mu},
		\end{align*}
		where the first equality is by definition of $\HIC$, the second because $I( S_1; B_{out}B' | X Y)_{\rho_\mu} = 0$ by the Markov propery of protocols ($X$, $Y$ and $S_1$ are all classical here), the third one is by chain rule, and the inequality is by the fact that $S_1$ is classical and $H(S_1) = H(p)$.
		
		Moreover, since $S_1$ is a classical register when $S_2$ is traced out,
		$$
			I(X; B_{out}B' | Y S_1)_{\rho_\mu} 
			= p I(X; B_{out}B' | Y)_{\rho_{\mu_1} } + (1-p) I(X; B_{out}B' | Y)_{\rho_{\mu_2} }.
		$$
		Hence:
		$$
			\HIC_{A \rightarrow B}(\Pi, \mu) \leq H(p) + p \HIC_{A \rightarrow B}(\Pi, \mu_1) + (1-p) \HIC_{A \rightarrow B}(\Pi, \mu_2).
		$$
	\end{proof}
	
	Then, we conclude the proof of Lemma~\ref{lem:QICmu0}.
	
	\begin{proof}[Proof of Lemma~\ref{lem:QICmu0}]
		Let us denote $\mu_1$ the probability distribution with weight $1$ on input $(1,1)$. Then, we can write:
		$$
			\mu = (1-w) \mu_0 + w \mu_1, \qquad \mu_0 = (1-w) \mu_0 + w \mu_0.
		$$
		Hence, by Lemma~\ref{lem:quasi-conv-improved}:
		\begin{align*}
			\QIC(\Pi, \mu)
			&\leq (1-w) \QIC(\Pi, \mu_0) + w \QIC(\Pi, \mu_1) + H(w)
			\\
			&\leq (1-w) \QIC(\Pi, \mu_0)  + H(w)
			\\
			&\leq \QIC(\Pi, \mu_0) + H(w).
		\end{align*}
	\end{proof}



\section{The Various Notions of Information Cost}
\label{sec:inflow-B-table-IC}

\begin{table}[!h]
\begin{center}
{\tabulinesep=1.2mm
\begin{tabu}{|c|c|c|c|c|c|}
\hline
   $\QIC$ & Definition~\ref{prelim:def:QIC} (see~\cite{Tou15})
   \\ \hline
   $\CIC$ & Definition~\ref{charac:def:CIC} (see~\cite{KLLGR15})
   \\ \hline
   $\HIC$ & Definition~\ref{charac:def:HIC}
   \\ \hline
   $\CRIC$ & Definition~\ref{charac:def:CRIC}
   \\ \hline
   $\SCIC,\SCRIC,\SHIC$ & Definition~\ref{charac:def:SCIC-SCRIC-SHIC}
   \\ \hline
   $\HCIC,\HCRIC,\HHIC$ & Definition~\ref{charac:def:HCIC-HCRIC-HHIC}
   \\ \hline
   IC & Definition~\ref{forgetclassical:def:defIC}
   \\ \hline
   $\RIC$ & Definition~\ref{forgetclassical:def:defRIC}
   \\ \hline
\end{tabu}
}
\vspace{0.1cm}
\caption{The classical and quantum notions of information cost used in this article.}
\label{table1}
\end{center}
\end{table}


\end{document}